\documentclass[conference]{IEEEtran}

\ifCLASSINFOpdf
\else
\fi
\usepackage[utf8]{inputenc}
\usepackage{cite}
\usepackage{graphicx}
\usepackage{textcomp}
\usepackage{bmpsize}
\usepackage{xcolor}
\usepackage{caption,subcaption}
\usepackage{xspace}
\usepackage{diagbox}
\usepackage{tabularx}
\usepackage{float}
\usepackage{comment}
\usepackage{multicol,multirow}
\usepackage{tabularx}
\usepackage{array}
\usepackage{ragged2e}
\usepackage{algorithm}
\usepackage{algpseudocode}

\usepackage{amsthm}
\usepackage{listings}
\usepackage{textcomp}
\let\labelindent\relax
\let\omicron\relax
\usepackage{mathptmx}
\usepackage[bigdelims]{newtxmath}
\usepackage[ieeetran]{henry}
\usepackage{tabularx}
\usepackage{microtype}
\PassOptionsToPackage{activate={true,nocompatibility},final,tracking=true,kerning=true,spacing=true,factor=1100,stretch=10,shrink=10}{microtype}
\microtypecontext{spacing=nonfrench}

\usepackage[most]{tcolorbox}

\usepackage{versions}
    \includeversion{extended}

\usepackage{tikz}
\usepackage{booktabs,multirow}

\newtheorem{theorem}{Theorem}[section]

\newtheorem{corollary}[theorem]{Corollary}

\newtheorem{definition}[theorem]{Definition}
\newtheorem{observation}[theorem]{Observation}

\usepackage{setspace}

\begin{document}
\title{Private Aggregate Queries to Untrusted Databases*}



%
\author[Syed Mahbub Hafiz and Chitrabhanu Gupta and Warren Wnuck and Brijesh Vora and Chen-Nee Chuah]
{
    \IEEEauthorblockN{\textlarger[0.5]{Syed Mahbub Hafiz}\textsuperscript{\textsection}\IEEEauthorrefmark{2}}
    \IEEEauthorblockA{University of California, Davis\\Indiana University, Bloomington\\
        shafiz@\{ucdavis.edu, indiana.edu\}
    }
    \and
    \IEEEauthorblockN{\textlarger[0.5]{Chitrabhanu Gupta\IEEEauthorrefmark{2}, Warren Wnuck, Brijesh Vora, and Chen-Nee Chuah}}
    \IEEEauthorblockA{Department of Electrical and Computer Engineering\\University of California, Davis, CA, USA\\
        \{cbgupta, wrwnuck, bhvora, chuah\}@ucdavis.edu
    }
}




\maketitle
\begingroup\renewcommand\thefootnote{\textsection}
\footnotetext{~This author worked on the theoretical and practical aspects of this project when he was affiliated with IU Bloomington and UC Davis, respectively.\\[.35ex]~~\IEEEauthorrefmark{2}Co-first author.\\
[.35ex]*This is a preprint version of the paper accepted at the Network and Distributed System Security (NDSS) Symposium 2024.}
\endgroup

\begin{abstract}
An essential part of ensuring privacy for internet service users is to protect what data they access so that the database host cannot infer sensitive information (e.g., political affiliation, sexual orientation, etc.) from the query pattern to exploit it or share it with third parties. Often, database users submit aggregate queries (e.g., SUM, MEAN, etc.) with searching and filtering constraints to extract statistically meaningful information from a database by seeking the privacy of its query's sensitive values and database interactions. Private information retrieval (PIR), a privacy-preserving cryptographic tool, solves a simplified version of this problem by hiding the database item that a client accesses. Most PIR protocols require the client to know the exact row index of the intended database item, which cannot support the complicated aggregation-based statistical query in a similar setting. Some works in the PIR space contain keyword searching and SQL-like queries, but most need multiple interactions between the PIR client and PIR servers. 
Some schemes support searching SQL-like expressive queries in a single round but fail to enable aggregate queries. These schemes are the main focus of this paper. To bridge the gap, we have built a general-purpose novel information-theoretic PIR (IT-PIR) framework that permits a user to fetch the aggregated result, hiding all sensitive sections of the complex query from the hosting PIR server in a single round of interaction. In other words, the server will not know which records contribute to the aggregation.
We then evaluate the feasibility of our protocol for both benchmarking and real-world application settings. For instance, in a complex aggregate query to the \emph{Twitter} microblogging database of $1$ million tweets, our protocol takes $0.014$ seconds for a PIR server to generate the result when the user is interested in one of $\sim\!3k$ user handles. In contrast, for a much-simplified task, not an aggregate but a positional query, Goldberg's regular IT-PIR (Oakland 2007) takes $1.13$ seconds. For all possible user handles, $300k$, it takes equal time compared to the regular IT-PIR. This example shows that complicated aggregate queries through our framework do not incur additional overhead if not less, compared to the conventional query.

\end{abstract}



\section{Introduction}\label{s:intro}
\vspace*{-0.1cm}


With the progressive digitization of information, concerns about how applications and software services keep one's information private keep increasing. Companies utilize their users' information for their gain, while employees within companies are capable of obtaining data for malicious activities~\cite{isr,homeland}. As the breach of users' privacy in digital space becomes more rampant, the need for privacy-enhancing technologies becomes more pressing. We specifically address the situation where a user wishes to retrieve information from an untrusted database without revealing to the database what specific information they wish to obtain. Private information retrieval (PIR) \cite{chor1998private, riise2019introduction, devet2012optimally, hafiz2019bit} is a cryptographic technique that solves this problem. However, while numerous privacy-enhancing technologies, including various PIR protocols, have been conceived over the past two decades to address this problem, the vast majority of them are not viable for facilitating a comprehensive set of statistical queries that might be required for performing data analytics on sensitive databases. A primary reason for this is that most protocols cannot support commonly used aggregate queries with a provable privacy guarantee. We develop a novel framework for information-theoretic private information retrieval protocols – that utilizes the concept of introducing sparse matrix-like auxiliary data structures similar to Hafiz-Henry's Querying for Queries work~\cite{hafiz2017querying} to support various aggregate queries frequently used in data analytics workflows. Unlike the vast majority of existing protocols, our scheme does not require physical positional information to fetch aggregated results from a database, and it just requires contextual information about the data to fetch the required blocks of data.

There are various practical applications of a protocol that can allow users to privately fetch data from untrusted databases and perform statistical queries on the same for data analysis. Some sample events for private information retrieval with aggregate queries include the following:

\emph{Event 1: Social Networking Platforms:} Social network platform databases are often used for various kinds of data analyses and thus often respond to various aggregate queries. However, untrusted social media platform databases can potentially learn substantial information about a user querying the database. A typical example of this would be identifying a user's political affiliation or interests by observing their queries. A user might be interested in the total number of positive reactions to social media posts made by a politician in interest, in which case they would send a SQL query to the database of the form: {\fontencoding{T1}\fontfamily{cmr}\selectfont \color{gray!75!black} SELECT SUM(number\_of\_likes) FROM user\_posts WHERE user\_id =`Joe Biden.'}
\normalfont

A user might also be curious about the social media presence of various politicians and construct a histogram of the number of social media posts made by politicians from a specific party. The equivalent SQL query would be: 
 {\fontencoding{T1}\fontfamily{cmr}\selectfont \color{gray!75!black} SELECT COUNT(*) FROM user\_posts WHERE party = `Democratic' GROUP BY user\_id.}
\normalfont
Privatizing these queries would prevent social media databases from identifying personal information about users through their search interests. 

\emph{Event 2: Booking Flights:} When browsing flight options for travel, users of a travel website or app are often likely to query for round-trip ticket prices. Additionally, some travel itineraries might also require multi-city flights with layovers, in which case the prices of all connecting flights are aggregated. These would require a query such as 
{\fontfamily{cmr}\selectfont\color{gray!75!black}\small SELECT SUM(price) FROM flights WHERE flight\_id in (`1003', `2319') }\normalfont
A private version of this sum query would prevent the database owner from gaining information about the user's travel plans and increasing flight prices on successive searches, and would also protect the user from targeted advertising. 

\emph{Event 3: Stock Market Data:} Databases containing information about stocks and the fluctuations of their values might be queried by several users interested in stock market trends. Users might be curious about the maximum daily fluctuation of a stock's value or mean fluctuation of a stock's value across a fixed period. The queries would be of the form: {\fontfamily{cmr}\selectfont \color{gray!75!black}\small{SELECT MAX(daily\_value\_change), stock\_id FROM stocks WHERE month = `June' GROUP BY stock\_id,}} \normalfont
or
{\fontfamily{cmr}\selectfont \color{gray!75!black}\small{SELECT AVG(daily\_value\_change) FROM stocks  WHERE month = `June.'}} \normalfont

The database owner that can track user queries would be able to observe a spike in interest toward a specific stock if many users execute the above queries over a short period and could misuse that information. Thus, such databases should ideally have PIR systems implemented.


The applications of our protocol discussed thus far all have the common assumption that the database is hosted within the infrastructure of the organization that owns the database. The discussions all focus on the various ways in which the entities that own the databases can potentially exploit the knowledge of the details of user queries. However, due to the volume of data organizations possess, it is common practice to outsource the databases in cloud infrastructures that are provided by third-party organizations. In such scenarios, even if the database owner does not intend to observe (and essentially abuse) user queries, the cloud service provider hosting the database might turn out to be malicious and be interested in learning the details of queries made to the databases. They could potentially monitor user queries and compile adequate information about users for targeted advertising or sell the compiled information to other organizations without the consent of the users. Our protocol could be implemented on the cloud infrastructure providers' servers to prevent misuse of user information and breach of privacy. 

A large number of the existing protocols for querying databases are not capable of performing keyword-based expressive queries. Even for the ones that do, they do not support aggregate queries. Furthermore, existing protocols capable of performing private aggregate queries are either limited by their complexity, such as in terms of the number of times a database is accessed for a single query, or the variety of aggregations they can support \cite{wang2017splinter, boyle2015function, zhao2022information, ahmad2021coeus, gui2021rethinking}. Thus, while performing data analytics, if a user wishes to not reveal any information about their queries, they face many levels of challenges, where only a few existing protocols can be successfully implemented to support keyword-based queries that are intuitive parallels of commonly used SQL queries. A data analyst may desire to have access to a full spectrum of tools and statistics required to develop data analytics pipelines without revealing any information to the database owner. Our key contributions are summarized as follows:
\vspace*{-0.15cm}
\begin{enumerate}
    \item We propose a framework that enables users to submit aggregate statistical queries on untrusted databases privately with a provable security guarantee, i.e., with PIR guarantees.
    \item To build the novel framework, we introduce a new kind of standard aggregate vector, unlike the basis vector known in the PIR literature. We propose the standard aggregate vectors that contain multiple ones instead of a single one to allow component-wise aggregation.
    \item The door of various possibilities is opened when we construct and batch auxiliary indexes of standard aggregate queries-based matrices utilizing querying for queries and polynomial batch coding techniques, respectively. This empowers the conflux of aggregation with searching, sorting, ranking, and various expressive queries.
    \item Having parallelizable computations, the GPU implementation of extra computations demonstrates remarkable benchmarking results.
    \item Additionally, we demonstrate through compelling case studies the practicality of implementing our protocol in real-world applications. We demonstrate that we can perform an aggregate query on a database with a million records and around $3,000$ search terms in just $0.014$ seconds. We have an implementation that can be scaled to enable aggregate queries to a database with $16$ million records and by $65$ thousand search terms, while the PIR server can serve up to $4$ thousand PIR clients per second. 
    \item Our artifact has been thoroughly evaluated and acquired `available,' `functional,' and `reproduced' badges. It is under an open-source license and available at \url{https://doi.org/10.5281/zenodo.10225325}~\cite{hafiz2022paq} and \url{https://github.com/smhafiz/private_queries_it_pir/tree/v1.0.0}. 
\end{enumerate}

\vspace*{-0.30cm}
\section{Threat Model}\label{s:threat-model}
\vspace*{-0.10cm}
We assume a threat model where the data provider is not trusted, which implies that the data provider can observe all queries run on their database by any user, the computations taking place on the server, and which database rows are scanned and which are left untouched. This applies to all the servers involved in the protocol which host the database. The users are any individual with access to the database for querying and fetching data from it. Under this assumption, we require the database provider not to be able to learn any contextual information about the queries that a user wishes to run on the database. Our protocol aims to mitigate passive activities such as eavesdropping. The user needs to know minimal information about the database, the number of rows in the permutation matrix data structures discussed in more detail in the later sections, and a keyword for the query is sufficient for privately retrieving the information from the database that the user desires. Another assumption for our threat model is that only a certain threshold ($\t$) number of PIR servers may collude, with information about the user’s query becoming revealed upon collusion of several servers greater than the specified threshold, $\t$. However, it is possible to tune this threshold as desired during the implementation of the protocol. Note that all further mentions of a database being untrusted would conform with the definition of untrusted provided at the beginning of this section. Additional details about expanding our model to a $\v$-Byzantine robust model can be found in Section~\ref{ss:byzantinerobust}.

\begin{figure}[h]
\begin{center}
\includegraphics[width=.48\textwidth]{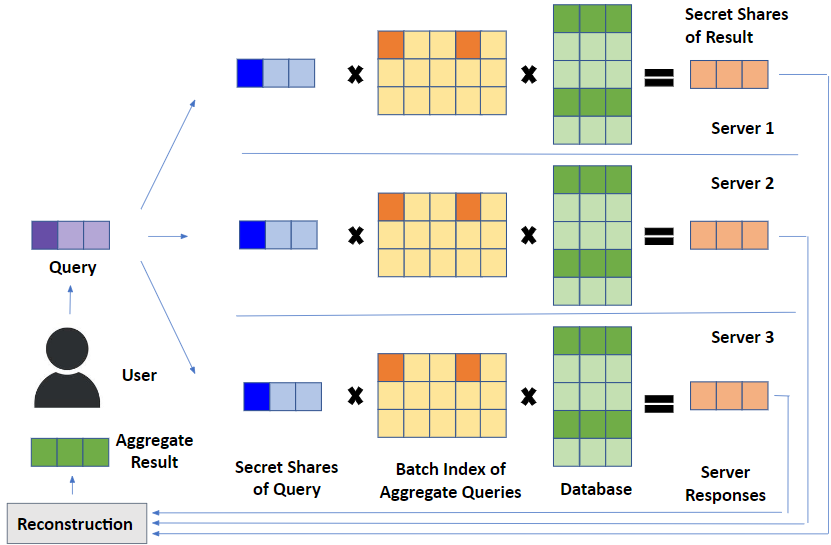}
\caption{Schematic diagram of the proposed indexes of aggregate queries-powered IT-PIR protocol: In order to demonstrate how our key protocol works, we highlight the example of a query to a hypothetical flight records database for finding the sum of ticket prices of connecting flights. To serve this query, the query vector is first created by the client and then Shamir’s Secret Sharing algorithm is used to generate a secret share vector for each of the servers. Each secret share vector is sent to a server, where they are multiplied with the relevant index of aggregate queries matrix, where each row corresponds to a different flight id. The resultant vector is then multiplied with the copy of the flight records database matrix on that server and the product is returned to the client. Each of these server responses are then used to reconstruct the aggregated response using the Lagrange polynomial interpolation. }
\label{fig:schematic}
\vspace*{-1cm}
\end{center}
\end{figure}

\section{Background}\label{s:background}
\vspace*{-0.10cm}
Privacy-enhancing technologies (PETs) \cite{goldberg1997privacy, tavani2001privacy, senivcar2003privacy} allow people to control the distribution and usage of information about them. PETs typically utilize cryptographic techniques with specific security properties that lead to their privacy guarantees. These cryptographic primitives’ security properties stem from fundamental conjectures of information theory or computation complexity. PETs utilizing computational complexity assume that it is computationally infeasible for an adversary to obtain sensitive information by observation or participation. Still, malicious parties with unlimited computation resources can break this system. On the other hand, adversaries do not threaten information-theoretic PIR (IT-PIR) protocols with infinite computation resources. However, they come with other assumptions, e.g., a widespread one relevant to our work is that out of all the parties participating in the protocol, only a certain number of them (up to a specific threshold number, $\t$) may collude. This assumption is observed in various kinds of protocols, such as secret sharing \cite{shamir1979share, beimel2011secret}, onion routing networks \cite{dingledine2004tor}, mix networks \cite{danezis2003mixminion}, and cryptographic voting \cite{ryan2006pret, chaum2009scantegrity}.

Multi-server IT-PIR is quite efficient compared to counterparts such as computational assumption-based single-server PIR, which suffers heavy computational overhead for (homomorphic) encryption. A naive single-server PIR involves the client downloading the entire database to retrieve the desired record and is thus infeasible for practical implementation. An improvement in the performance and feasibility of deployment for single server PIR schemes have been observed with the introduction of the offline-online PIRs, where during the offline phase, some information about the database is precomputed and exchanged to optimize the server’s processing time during the online phase. However, even the most performant single-server PIR schemes~\cite{henzinger2023one, menon2022spiral} suffer from high query processing times and, thus, are overall less efficient than multi-server schemes.

Although more efficient, successful practical implementation of multi-server IT-PIR schemes requires a non-collusion guarantee between the servers, as stated earlier. Some works explicate how non-collusion assumptions in cryptographic protocols such as PIR and secure multi-party computation are deployable in practice~\cite{gong2022more, wang2016information, wang2014efficient}. These works offer collusion mitigation strategies for collusions that do not occur via unknown side-channels external to the protocols. They propose modeling collusion as games and design retaliation mechanisms that result in the most viable stable strategies for the game to not involve collusion, thus leading to the adoption of these strategies within the protocols.

IT-PIR assumes that the database host is untrusted; thus, a user’s query must not indicate the requested data. The trivial solution to this problem entails sharing the entire database with the user so that they may obtain their desired rows of data locally. However, as discussed, this is obviously not practical because databases can be quite large. So we instead explore protocols with total communication costs less than (sublinear to) the database size, as stated in~\cite{hafiz2019bit}, as the \textit{non-triviality} property of any practical PIR protocols.

The roots of our proposed protocol can be traced back to Chor et al.’s multi-server IT-PIR protocol~\cite{chor1998private} based on the idea of sharing queries among multiple non-colluding database servers. The protocol models the database as a string of bits out of which a user fetches a specific bit while keeping the index of the fetched bit private and is further extended in Goldberg’s robust IT-PIR protocol \cite{goldberg2007improving}, which models the database as a set of $b$-bit blocks and allows a user to fetch multiple such blocks obliviously. This emulates a more realistic scenario and led to the development of protocols that facilitate querying variable-length records \cite{henry2013one}.
\vspace*{-0.25cm}
\subsection{Vector-Matrix PIR Model}\label{s:vector-matrix-model}
\vspace*{-0.15cm}
Our protocol is built around the commonly used vector-matrix model for PIR \cite{goldberg2007improving}, and we model our database $\D$ as an $\rbys$ matrix, $\D\in\F^{\rbys}$, where $\r$ corresponds to the number of data blocks. Each data block has $\s$ words, $\Di\in\F^{\s}$, and each word $\w$ is a finite field element from $\F$. To fetch the \ith block of data, $\Di$, a user encodes an $\r$-dimensional query vector, $\e\in\F^{\r}$, with a $1$ in the \ith position and $0$s at every other index. Per linear algebra, the product of this query vector with the database matrix, $\e\cdot\D\in\F^{\s}$, produces the desired \ith data block, $\Di$. These requests are positional queries because they require a user to know the physical location of the blocks of data in the database, but having to possess knowledge of row numbers for every block in the matrix can be a major inconvenience for both the user and the service provider. This procedure is not private and thus requires strategic modification to make it private. As mentioned earlier, we opt to use linear secret sharing to overcome this problem, where the user shares their query vector component-wise across $\ell$ number of servers, the share vectors are multiplied with copies of the database matrix hosted in each server, and the user receives independent products from each of the servers. The user then performs a component-wise secret reconstruction using the responses received from the servers to obtain the desired block of data. Following in the footsteps of \cite{hafiz2017querying}, we select Shamir’s secret sharing scheme \cite{shamir1979share} as proposed in Goldberg’s IT-PIR protocol \cite{goldberg2007improving} as our choice of linear secret sharing scheme. Shamir’s $(\t+1, \ell)$-threshold scheme in the vector-matrix model ensures that the user obtains their desired results as long as $\t+1$ or more servers out of $\ell$ respond. In contrast, a collusion of fewer than $\t+1$ servers fails to reveal anything about the information requested by the user.

Our protocol exploits the concept of indexes of queries introduced in Hafiz-Henry's protocol. The protocol introduced in \cite{hafiz2017querying} utilizes indexes of queries to support expressive queries that utilize contextual information rather than positional information to fetch relevant data blocks. Although Chor et al.~\cite{chor1997private} introduced a scheme for keyword-based querying that translates keyword searches into positional queries, is comparable but with multiple rounds of interactions, \cite{hafiz2017querying}’s index of queries mechanism was more suitable for our protocol since it requires a single round of interaction to fetch the data, thus minimizing communication costs.



\vspace*{-0.35cm}
\subsection{Polynomial Batch Coding}\label{s:batch-coding}
\vspace*{-0.15cm}
To potentially reduce the size of the database hosted on a server, Henry~\cite{henry2016polynomial} modified the vector-matrix model, which results in each server hosting an encoded bucket obtained from the database. The encoded bucket is typically smaller than the database, reducing the cost of hosting data and communication and computation costs.

As proposed in \cite{henry2016polynomial}, the buckets are obtained via the ramification of Shamir’s secret-sharing scheme. In the typical secret-sharing scheme, there would be a threshold $\t$, the maximum number of shares an attacker could possess and still fail to learn the secret. But possessing $\t+1$ or more shares results in the secret being revealed to the attacker. Rampification involves relaxing the threshold, where the privacy does not entirely collapse as soon as an attacker comes to possess $\t+1$ shares. Instead, while the attacker still knows nothing about the secret while possessing $\t$ or fewer shares, to learn the secret thoroughly, the attacker would now have to possess $\t+\u$ secret shares. Possessing more than $\t$ but fewer than $\t+\u$ shares results in a partial loss of privacy. This relaxation benefits a protocol capable of packing $\u$-times as many bits into the secret sharing scheme. The shares possessed by the user still look the same, but there is a $\u$-fold improvement in the information that can be packed. So, in Shamir’s secret sharing, the random polynomials still encode a secret in the $y$-intercept. The difference is that we can obtain different secrets upon evaluating the polynomials at other values of $x$.

This idea is extended to modify the database by changing its $\r$ blocks into $\frac{\r}{\u}$ $\u$-tuples, interpolating component-wise through each of them at some predefined $x$-coordinates to obtain $\frac{\r}{\u}$ length-$\s$ vectors of degree-$(\u - 1)$ polynomials and then placing a single component-wise evaluation of each vector of polynomials into each of the $\ell > \u$ buckets. So, to fetch the \ith block from the database, the user would identify which out of the available $\frac{\r}{\u}$ buckets possesses evaluations of the polynomial vector passing through the block being sought and the $x$-coordinate value at which the polynomial passes through the block. 

In the simplest form of our proposed scheme, every PIR server must host copies of a potentially large number of indexes of aggregate queries to cater to a broad spectrum of user queries. However, an adversary capable of observing the activity in the servers can identify which index of queries is accessed by a user query, thus obtaining information about the user’s query and subsequent data interests. This could be mitigated if, instead of having multiple indexes of queries for various queries, all the indexes of queries that need to be hosted on the server could be batched together using polynomial batch coding, thus resulting in just a single batch index of queries of polynomials. As a result, every user query must pass through the same batch index of queries hosted on each server, thus obstructing information leakage about the user query through the hosted indexes of queries. Our protocol greatly benefits from $\u$-ary coding because the indexes of queries stored in the servers can be batched together, preventing information leakage, and drastically reducing the storage and communication cost and the cost of computations when queried.

Further enhancements to the scheme are achieved by implementing the data structure called indexes of batch queries, where users can aggregate multiple data blocks through a single request. Like indexes of queries, the concept of $\u$-ary coding can also be used to perform polynomial batch coding of multiple users queries into a single query vector. The responses provided by the servers after the batched queries pass through indexes of queries can be reconstructed and evaluated at different $x$-coordinate values to obtain the blocks requested by each query that was batched into the query vector. While this enhancement allows for the performance of top-$K$ queries in its simplest instantiation, it can also be utilized by a user to generate \emph{histograms} through a single request by batching together count queries for every value of the histogram field being constructed. The mechanism and practical examples of histogram queries are further explored in later sections discussing implementation details of supported aggregate queries and case studies.
\vspace*{-0.25cm}
\subsection{Index of Queries Technique}\label{s:index-of-queries}
\vspace*{-0.15cm}
A data structure called a permutation matrix is at the heart of the idea of indexes of queries, as introduced in \cite{hafiz2017querying}. Permutation matrices are constructed by permuting the rows of an identity matrix. In general, the permutation is performed according to the context of various slices of information. In permutation matrices, each row is mapped to a possible column value in the database to apply a query filter. Each column in the permutation matrix corresponds to a row in the database. To serve a query for a specific condition, for each row of the matrix, a $1$ is placed in the required column of the matrix to demarcate the corresponding row in the database, thus indicating which record in the database satisfies the condition for each of the possible values of the filter column in the database (each of which is mapped to a row in the permutation matrix). Thus, a permutation matrix for the number of hospitalizations and one for the number of instances of a particular disease would look different. Still, it can be constructed from the same database containing all the relevant information for constructing the matrices. Depending on the $1$s in the permutation matrices, query filter clauses involving equalities and inequalities can be supported. 

Permutation matrices are the most uncomplicated instances of indexes of queries, which are a generalization of the same. Permutation matrices, by definition, are restricted to a single $1$ in each row and column, thus representing the particular case of the index of matrices, where a sorted view of each block of the database is presented. The index of queries is a broader concept. The matrices may take a few different forms, possessing no non-zero elements in some columns, making certain database sections inaccessible through them, or even having multiple $1$s in each row. We discover that permuting identity matrices limit the queries to keyword searches and many other non-aggregate applications; permuting over a matrix with multiple $1$s in each row and $0$s in the remaining positions allows for each row in the matrix, aggregation over several records in the database, the columns in the permutation matrix corresponding to which have $1$s in that row, owing to linearity. Thus, matrices permuted from identity matrices and matrices with multiple $1$s in each row are constructed according to query interests for storage on the PIR servers.

The number of columns in each index of queries matches the number of rows in the database. Still, a user is only required to know the number of rows in each index of queries to construct a query since the dimension of the query vector must be the same as the number of rows in the index of queries. Additionally, a user does not require any positional knowledge for the information they seek, contrary to most prior PIR schemes, thus also providing the database owner with the benefit of less sharing of information. To perform a query successfully and obtain the information that they desire, a user needs to know the number of rows of the index of queries to construct a correct dimensional query vector and the necessary keyword to pass along to the PIR servers to identify the correct index of queries for the query, for example, ``hospitalizations'' to locate the index of queries tracking the number of hospitalizations and obtain information about hospitalizations by state.
\vspace*{-0.25cm}
\section{Private Aggregate Queries}\label{s:aggregate-queries}
\vspace*{-0.10cm}

The non-private version of the IT-PIR framework is quite simple, the multiplication between the standard basis vector (one-hot vector) and the database produces the result. However, the application of the XOR-based additive secret sharing technique or Shamir's additive secret sharing turns it into a private version with perfect (information-theoretic) security, where per-server computation is much better than computational-PIR techniques. Likewise, though the multiplication between a standard aggregate vector and the database produces component-wise addition of columns, combining this with batch coding and the index of queries technique empowers PIR to support more sophisticated queries with the same perfect security guarantee and better communication and computation costs (when p << r). This is true that all of these are achievable due to the linearity of the PIR constructions. In this section, we present our indexes of aggregate queries techniques-based IT-PIR constructions. 
\vspace*{-0.25cm}
\subsection{Indexes of aggregate queries}\label{s:expressive:aggregate}
\vspace*{-0.15cm}
First, we define the standard aggregate vector and simple index of aggregate queries next.

\begin{definition}\label{def:aggregatebasis}
Let $ \setofindexes$ be a set of indices, $\numberofones$ be the cardinality of the set $\setofindexes$, $\numberofones=\cardinality{\setofindexes}$, and $2\leq\numberofones\leq\rows$. A \textsl{standard aggregate vector} $ \aggregate[\setofindexes]\in\F^{\rows}$ is a $(0,1)$-vector with an $1$ value at each index of the set $\setofindexes$.
\end{definition}
\vspace*{-0.20cm}
Note that when $\cardinality{\setofindexes}=1$, the standard aggregate vector turns into a standard basis vector. The following observation is straightforward when we multiply a standard aggregate vector with a database using simple linear algebra.

\begin{observation}
If $\aggregate[\setofindexes]$ is a standard aggregate vector and $\setofindexes$ is a set of indices $\set{i_1,\ldots,i_{\lastone}}$ where $\lastone=\hammingweight{\aggregate[\setofindexes]}$,\footnote{$i_{\lastone}$ is the position where the last $1$ is in the $\aggregate[\setofindexes]$, and $\lastone$ can be computed by the hamming weight, $\hammingweight{}$, of $\aggregate[\setofindexes]$.} then $\aggregate[\setofindexes]\cdot\D$ is a component-wise linear combination of particular records (indexed by $\setofindexes$) of the database.
\end{observation}
\vspace*{-0.20cm}
That means, $\aggregate[\setofindexes]\cdot\D=\Di[i_1]+\ldots+\Di[i_{\lastone}]\in\F^{\cols}$, is a component-wise addition across selected records. We can accumulate such standard aggregate vectors and build an index of aggregate queries. Next, we define the simple index of aggregate queries in light of the simple index of queries~\cite{hafiz2017querying}.

\begin{definition}\label{def:simpleindex:aggregate}
A \textsl{simple index of aggregate queries} for a database $\D\in\F^{\rbys}$ is a $(0\mkern1mu,\mkern-1mu1)\mkern2mu$-matrix $\Pi\in\F^{\pbyr}$ in which each row is a \textsl{standard aggregate vector}.
\end{definition}
\vspace*{-0.20cm}
Now, we can quickly realize the following observations\emdash

\begin{observation}\label{obs:permute:aggregate}
If $\e[]\in\F^{\r}$ is a standard basis vector and $\Pi\in\F^{\rbyr}$ is a simple index of aggregate queries, then $\e[]\cdot\Pi$ is a standard aggregate vector $\aggregate[\setofindexes]\in\F^{\rows}$, namely, $\aggregate[\setofindexes]=\e[]\cdot\Pi$.
\end{observation}

\begin{observation}\label{obs:pseudopermute:aggregate}
Let \smash{$\e[]\in\F^{\p}$} be a standard basic vector and let $\Pi\in\F^{\pbyr}$ be a simple index of aggregate queries with rank $\p \lastone$.
If \smash{$(\x[1],\vecQ[1]),\ldots,(\x[\ell],\vecQ[\ell])$} is a component-wise $(\t+1,\ell)\mkern2mu$-threshold sharing of $\e[]$, then \smash{$(\x[1],\vecQ[1]\cdot\Pi),\ldots,(\x[\ell],\vecQ[\ell]\cdot\Pi)$} is a component-wise $(\t+1,\ell)\mkern2mu$- and $(1,\ell)\mkern2mu$-threshold sharing of a standard aggregate vector $\aggregate[\setofindexes]\in\F^{\r}$; namely, of $\aggregate[\setofindexes]=\e[]\cdot\Pi$.\footnote{Specifically, the \smash{$\r-\p \lastone$} entries corresponding to all-\smash{$0$} columns in \smash{$\Pi$} are \smash{$(1,\ell)$}-threshold shares of \smash{$0$}; the remaining \smash{$\p \lastone$} entries are each \smash{$(\t+1,\ell)$}-threshold shares of either \smash{$0$} or \smash{$1$}.}
\end{observation}
\vspace*{-0.20cm}
We define the $t$-privacy for indexes of aggregate queries as follows\emdash

\begin{definition}\label{def:tprivacy:aggregate}
Let $\D\in\F^{\rbys}$ and let each $\oldPi_1,\ldots,\oldPi_n$ be an index of aggregate queries\footnote{We didn't use ``simple'' here because each \smash{$\oldPi_{\j}$} can either be a simple index of aggregate queries or one of the more sophisticated types we introduce next. In particular, by permitting \smash{$\oldPi_{\j}$} to be various indexes of aggregate queries, we can use \definitionref{def:tprivacy:aggregate} to define privacy for all constructions.} for $\D$. Requests are \textsl{$\t$-private with respect to $\oldPi_1,\ldots,\oldPi_n$} if, for every coalition $\coalition\subseteq\ival{\ell}$ of at most $\t$ servers, for every record index $i\in\ival{\p}$, and for every index of aggregate queries $\oldPi\in\{\oldPi_1,\ldots,\oldPi_n\}$,
\begin{align*}
    \Pr\bigl[I=i\bigm\vert \Query[\coalition]=(\oldPi\semicolon\vecQ[\j_1],\ldots,\vecQ[\j_{\t}])\bigr]&=\Pr\bigl[I=i\bigm\vert E_{\oldPi}\bigr],
\end{align*}
where $I$ and $\Query[\coalition]$ indicate the random variables respectively describing the ``category'' index the user requests and the combined distribution of query vectors it sends to servers in $\coalition$ (including the ``hint'' that the query would go through $\oldPi$), and where $E_{\oldPi}$ is the event that the request is passed through $\oldPi$.
\end{definition}

From this definition, we understand that we can replace any ``index of queries'' matrices in any observations, definitions, and theories with the ``index of aggregate queries'' that proposed in~\cite{hafiz2017querying}. For example, we have\emdash

\begin{theorem}\label{THM:UARY:aggregate}
Fix $\u>1$ and $j\in\ival[0]{\u-1}$, and let $\oldPi=\bigl(\oldPi_1,\ldots,\oldPi_{\ell}\bigr)\in\bigl(\F^{\pbyr}\bigr){}^{\ell}$ be $\ell$ buckets of a $\u$-batch index of aggregate queries with bucket (server) coordinates $\x[1],\ldots,\x[\ell]\in\F\setminus\{0,\ldots,\u-1\}$. If \smash{$(\x[1],\vecQ[\j_1]),\ldots,(\x[\ell],\vecQ[\j_{\ell}])$} is a sequence of component-wise $(\t+1,\ell)\mkern2mu$-threshold shares of a standard basis vector $\e[]\in\F^{\p}$ encoded at $x=j$, then \smash{$(\oldPi,\x[1],\vecQ[\j_1]),\ldots,(\oldPi,\x[\ell],\vecQ[\j_{\ell}])$} is $\t$-private with respect to $\oldPi$.
\end{theorem}
\vspace*{-0.25cm}

Where we define $\u$-batch index of aggregate queries as follows\emdash

\begin{definition}\label{def:batchindex:aggregate}
Say $\u>1$ and let $\x[1],\ldots,\x[\ell]\in\F\setminus\{0\mkern0.5mu,\ldots,\u-1\}$ be pairwise distinct scalars. 
A sequence $\oldPi_1,\ldots,\oldPi_{\ell}\in\F^{\pbyr}$ of matrices is a \textsl{$\u$-batch index of aggregate queries} for Goldberg's IT-PIR with bucket coordinates $\x[1],\ldots,\x[\ell]$ if (i)~$\oldPi_{i_1}\neq\oldPi_{i_2}$ for some $i_1,i_2\in\ival{\ell}$, and (ii)~for each $j=0\mkern1.5mu,\ldots,\u-1$,
\begin{align*}
\pi_j\coloneqq\nsum[1.35]_{i=1}^{\ell}\oldPi_i\cdot
\text{\larger[0.5]$\displaystyle\bigl(\tfrac{j\mkern3mu-\mkern3mux_1}{x_i\mkern3mu-\mkern3mux_1}\bigr)\cdots\bigl(\tfrac{j\mkern3mu-\mkern3mux_{i-1}}{x_i\mkern3mu-\mkern3mux_{i-1}}\bigr)\bigl(\tfrac{j\mkern3mu-\mkern3mux_{i+1}}{x_i\mkern3mu-\mkern3mux_{i+1}}\bigr)\cdots\bigl(\tfrac{j\mkern3mu-\mkern3mux_{\ell}}{x_i\mkern3mu-\mkern3mux_{\ell}}\bigr)$}
\end{align*}
is a simple index of aggregate queries.
\end{definition}

\begin{figure}[t]
\begin{center}
\begin{align*}\tabcolsep=4pt\phantom{\D\coloneqq}&\left.~{\color{black}
\scalebox{0.5}{\begin{tabular}{p{3cm}p{1.8cm}p{1.8cm}p{1.5cm}p{2.8cm}p{3cm}}
\textsmaller{\texttt{Hospitalization\_ID}}&
\textsmaller{\texttt{Patient\_ID}}&
\textsmaller{\texttt{Admit\_date}}&
\textsmaller{\texttt{Gender\_ID}}&
\textsmaller{\texttt{Days\_hospitalized}}&
\textsmaller{\texttt{State\_ID}}\\
\end{tabular}}}\right.\\[-1.35ex]
\D\coloneqq&\left[
\scalebox{0.5}{\begin{tabular}{p{3cm}p{1.8cm}p{1.8cm}p{1.5cm}p{2.8cm}p{1.5cm}}
0001   & 1    &01-02-2022&   1 (Male)&10&2 (OR)\\
 0002&   2  & 01-04-2022   &1 (Male)& 2&1 (CA)\\
 0003 &3 & 08-06-2022&  2 (Female)&14&3 (WA)\\
 0004    &1 & 07-23-2022&  1 (Male)&2&2 (OR)\\
 0005&   3  & 09-01-2022&2 (Female)&7&3 (WA)\\
 0006& 4  & 05-14-2022   &3 (Other) &2&1 (CA)\\
\end{tabular}}\right]
\end{align*}
\vspace*{-0.25cm}\caption*{Small sample hospitalization record database.}\label{fig:toydatabase}
\end{center}
\vspace*{-0.65cm}
\end{figure}

\vspace*{-0.25cm}
\begin{algorithm}
\caption{General Algorithm For Our Protocol}\label{alg:cap}
\textbf{STEP 1:} Secret shares of client query vector generated\\
\textbf{STEP 2:} Secret shares of client query vector and keyword hint to select batch index of aggregate queries shared with servers\\
\textbf{STEP 3:} Keyword hint isolates required batch index of aggregate queries in each server \\
\textbf{STEP 4:} Secret share of client query vector multiplied with the selected batch index of aggregate queries in each server \\
\textbf{STEP 5:} Result obtained in the previous step multiplied with database matrix in each server \\
\textbf{STEP 6:} Results from each server sent back to the user \\
\textbf{STEP 7:} Results obtained from each server used to reconstruct desired query result \\
\end{algorithm}
\vspace*{-0.25cm}

\vspace*{-0.15cm}
\subsection{Various Indexes of Aggregate Queries}
\vspace*{-0.15cm}
The index of aggregate queries scheme facilitates the retrieval of results for several commonly used aggregate queries with a single interaction with the servers, thus making it a powerful tool to build an aggregate query framework around. It can support SUM, COUNT, MIN/MAX Histogram, and MEAN queries. Our scheme does not require the user to be aware of any details about the database being queried or the indexes of aggregate queries that are required for the aggregation (other than some semantic understanding for keyword-based indication of information to be aggregated). 
For each round of interaction (i.e., a client sends a query to the server, and the server responds back to it), the upload cost is O(p), (p is the number of search terms and the height of the index matrix), and the download cost is O(s) (s is the number of words of a record) regardless of the query type and how many records are retrieved. We highlight examples of aggregate queries using a sample hospitalization records database $\D$ as the target for the queries.

\subsubsection{Indexes of SUM Queries}\label{sss:sumquery}
To execute a SUM query in the form: {\fontfamily{cmr}\selectfont\color{gray!75!black}\small SELECT SUM(days\_hospitalized) FROM daily\_patient\_records WHERE patient\_id = 3. }\normalfont
An index of aggregate queries is constructed based on the daily\_patient\_records table that indicates each day of hospitalization for each possible value of A and is shared with each server containing a copy of the database. It is hosted by the server alongside indexes of aggregate queries for other potential searches as well. The client constructs an a-dimensional query vector, where a is the number of possible values for the variable in the filter clause. The scheme then follows the steps described in Algorithm 1 to generate the desired SUM. The index of aggregate queries matrix for serving this query is:
\vspace*{-0.10cm}
\begin{align*}
\Pi[patient]\coloneqq
\!\begin{array}{l}
\color{black}{\fontsize{7}{7}\selectfont \texttt{(Patient 1)}} \\
\color{black}{\fontsize{7}{7}\selectfont \texttt{(Patient 2)}} \\
\color{black}{\fontsize{7}{7}\selectfont \texttt{(Patient 3)}} \\
\color{black}{\fontsize{7}{7}\selectfont \texttt{(Patient 4)}}
\end{array}\!
\left[\!\begin{array}{llllll}
1&0&0&1&0&0\\
0&1&0&0&0&0\\
0&0&1&0&1&0\\
0&0&0&0&0&1
\end{array}\!\right]
\end{align*}
Similarly, to find the total number of days of hospitalization for male-identifying patients who were admitted before June $2022$, we would require an index of aggregate queries of the form:
\vspace*{-0.30cm}
\begin{align*}
\Pi[duration]\coloneqq 
\!\begin{array}{l}
\color{black}{\fontsize{7}{7}\selectfont \texttt{(Male)}} \\
\color{black}{\fontsize{7}{7}\selectfont \texttt{(Female)}} \\
\color{black}{\fontsize{7}{7}\selectfont \texttt{(Other)}}
\end{array}\!
\left[\!\begin{array}{llllll}
1&1&0&0&0&0\\
0&0&0&0&0&0\\
0&0&0&0&0&1
\end{array}\!\right]
\end{align*}

\vspace*{-0.25cm}
\subsubsection{Indexes of COUNT and Histogram Queries}\label{sss:counthistogramquery}
COUNT queries work similarly to SUM queries but might require some preprocessing or post-processing depending on the query's field. We start with a relatively simple example of counting the number of patients who identify as female. The SQL query would look something like:
{\fontfamily{cmr}\selectfont\color{gray!75!black}\small SELECT COUNT(*) FROM daily\_patient\_records WHERE gender\_id = G. }\normalfont

Typically, most databases are designed to have a constant integer, the value stored in the id column for a specific categorical variable, whose actual value is stored in a master table. Filtering on the id column and counting the rows would provide the desired response for the query. For this specific query, the client constructed a g-dimensional query vector, where g is the number of different possible values of gender id. The secret shares of the query vector for the gender id corresponding to females are sent to the different servers, along with a keyword indicating aggregation over gender. Each server hosts an index of aggregate queries associated with gender and a copy of the database. The secret shares filter the index of aggregate queries via multiplication, the product of their multiplication is then multiplied by the database to return shares from each server back to the user. Upon reconstruction, the user will receive an aggregation over the gender\_id column, which, when divided by the gender\_id corresponding to the desired gender, yields the count value, which in this case would be the count of hospitalized patients who identify as female. The following is the index of aggregate queries matrix for serving this query:
\vspace*{-0.15cm}
\begin{align*}
\Pi[population]\coloneqq 
\!\begin{array}{l}
\color{black}{\fontsize{7}{7}\selectfont \texttt{(Male)}} \\
\color{black}{\fontsize{7}{7}\selectfont \texttt{(Female)}} \\
\color{black}{\fontsize{7}{7}\selectfont \texttt{(Other)}}
\end{array}\!
\left[\!\begin{array}{llllll}
1&1&0&1&0&0\\
0&0&1&0&1&0\\
0&0&0&0&0&1
\end{array}\!\right]
\end{align*}
Some preprocessing is required on the database to extend the count operation to count using categorical filter variables. Categorical fields operated on by count queries should be identified, and a separate master table must be constructed with integer-type id columns for each categorical variable to be counted. These columns are populated with arbitrary values for each level of each variable. The master table can be used to generate a denormalized view on which the count query can operate. This preprocessing is required only if an id column does not exist for a categorical variable that needs to be counted. However, most databases are designed such that categorical variables of importance come with associated id columns for usage in fact tables. Consider the example of constructing a query to find the total number of hospitalizations in the state of California. For the execution of this query, if a state\_id column is unavailable in the database, a master table must be constructed where one of the columns is state\_id, and California is assigned an arbitrary integer value I. The SQL query executed on a denormalized view of the database would be: {\fontfamily{cmr}\selectfont\color{gray!75!black}\small SELECT COUNT(*) FROM hospital\_admission\_records WHERE state\_id = I. }\normalfont
The procedure for obtaining the result of this query is identical to that of the previous example, and follows Algorithm 1. The necessary index of aggregate queries matrix is:
\vspace*{-0.30cm}
\begin{align*}
\Pi[state]\coloneqq
\!\begin{array}{l}
\color{black}{\fontsize{7}{7}\selectfont \texttt{(California)}} \\
\color{black}{\fontsize{7}{7}\selectfont \texttt{(Oregon)}} \\
\color{black}{\fontsize{7}{7}\selectfont \texttt{(Washington)}}
\end{array}\!
\left[\!\begin{array}{llllll}
0&1&0&0&0&1\\
1&0&0&1&0&0\\
0&0&1&0&1&0
\end{array}\!\right]
\end{align*}
Histograms work similarly to count queries, with the difference being that histograms require counts of every level of a categorical variable rather than the count of a single value. This can be demonstrated with the help of the examples used for explaining the mechanism for the COUNT query. In the first example, instead of counting the number of hospitalized patients who identify as female, the user might instead want to construct a histogram of hospitalized patients by gender. The corresponding SQL query would be: {\fontfamily{cmr}\selectfont\color{gray!75!black}\small SELECT COUNT(*) FROM daily\_patient\_records GROUP BY gender\_id. }\normalfont

This would simply translate to a set of count queries, each counting a possible value of the gender\_id field. Suppose the number of different possible gender\_ids is g. In that case, g query vectors are batched together instead of just a single query vector, as in the case of obtaining the number of hospitalized patients who identify as female. The protocol then follows Algorithm 1 in order to compute the counts required to construct the histogram.

Similarly, if a histogram is required for a categorical variable with non-numeric values, the same preprocessing is required, as in the case of the second example for COUNT queries. The entire process is identical, except for the construction of the query vector. For the example of obtaining the histogram of hospitalized patients by state, the construction of the query vector would involve identifying the number of different possible values of state\_id and then batching together as many queries, one for each state\_id. The secret shares of this batched query vector would be shared with the servers along with a keyword suggesting that the target index of aggregate queries is the one associated with the state. This would serve the SQL query : {\fontfamily{cmr}\selectfont\color{gray!75!black}\small SELECT COUNT(*) FROM hospital\_admission\_records GROUP BY state\_id.}
\normalfont

\subsubsection{Indexes of MIN and MAX Queries:}
Min and Max aggregations can be achieved using simple indexes of queries. An example of a query for demonstrating the usage of MIN and MAX functions is a user request for the oldest date of hospitalization in California and, conversely, the most recent date of hospitalization in California. The corresponding SQL queries would be: {\fontfamily{cmr}\selectfont\color{gray!75!black}\small SELECT MIN(date\_time) FROM hospital\_admission\_records WHERE state\_id = I and SELECT MAX(date\_time) 
hospital\_admission\_records WHERE state\_id = I. }\normalfont

To use MIN or MAX in conjunction with GROUP BY, for instance, in a situation where the user is curious about the state with the fewest hospital admissions, they can obtain the response of a corresponding histogram query on the GROUP BY field and perform a sorting operation to yield the desired MIN or MAX.

The index of queries matrices required for serving the MAX and MIN SQL queries mentioned earlier are:
\begin{align*}
\Pi[latest admission]\coloneqq
\!\begin{array}{l}
\color{black}{\fontsize{7}{7}\selectfont \texttt{(California)}} \\
\color{black}{\fontsize{7}{7}\selectfont \texttt{(Oregon)}} \\
\color{black}{\fontsize{7}{7}\selectfont \texttt{(Washington)}}
\end{array}\!
\left[\!\begin{array}{llllll}
0&1&0&0&0&0\\
1&0&0&0&0&0\\
0&0&1&0&0&0
\end{array}\!\right]
\end{align*}
\vspace*{-0.35cm}
\begin{align*}
\Pi[oldest admission]\coloneqq
\!\begin{array}{l}
\color{black}{\fontsize{7}{7}\selectfont \texttt{(California)}} \\
\color{black}{\fontsize{7}{7}\selectfont \texttt{(Oregon)}} \\
\color{black}{\fontsize{7}{7}\selectfont \texttt{(Washington)}}
\end{array}\!
\left[\!\begin{array}{llllll}
0&0&0&0&0&1\\
0&0&0&1&0&0\\
0&0&0&0&1&0
\end{array}\!\right]
\end{align*}


\subsubsection{Indexes of Mean Queries:}
Calculating a mean requires computing a sum and a count for the same column. To retrieve the result of a mean query from the database, a user would have to batch [\hyperref[ss:batchingindexes]{4.3}] two queries, one each for computing a summation and the other for fetching the count, and the corresponding mechanisms are as described previously. For example, a user might be interested in the average age of hospitalized patients in California, and let us assume that the State\_ID for California in our example database is 1. The corresponding SQL query would be: {\fontfamily{cmr}\selectfont\color{gray!75!black}\small SELECT MEAN(age) FROM daily\_patient\_records WHERE state\_id = I.}
\normalfont 

To obtain this information, the query can be deconstructed into two queries: {\fontfamily{cmr}\selectfont\color{gray!75!black}\small SELECT SUM(age) FROM daily\_patient\_records WHERE state\_id = I 
 and SELECT COUNT(*) FROM daily\_patient\_records WHERE state\_id = I.}
\normalfont

In our example, I = 1. For both queries, an i-dimensional query vector must be constructed, where i is the number of possible values for the state\_id field, and the protocol would then proceed to perform the steps in Algorithm 1.
The response received from the servers for these two batched queries would provide the user with the sum of the ages of the patients hospitalized in California and the number of patients hospitalized in California. Thus, the user can now easily perform a local division operation with the response of the SUM query and that of the COUNT query to obtain the desired mean age of the patients hospitalized in California.
\vspace*{-0.35cm}

\vspace*{-.15cm}
\subsection{Batch Indexes of Aggregate Queries}\label{ss:batchingindexes}
\vspace*{-0.15cm}
Here, we showcase how we could batch $\u$ number of indexes of aggregate queries 
using the hospitalization database $\D$, which holds six records. Its schema includes the attributes of hospitalization ID, patient ID, admit date, gender ID, number of days hospitalized, and state ID. Suppose we target to build a PIR set up to support a user who is interested in either ``to find the total number of days of hospitalization for male-identifying patients who were admitted before June $2022$,'' or ``the count of hospitalized patients who identify as female.'' To facilitate such interesting (but complicated) aggregate queries, we recall two simple indexes of aggregate queries, $\Pi[duration]$ and $\Pi[population],$ described in~\sectionref{sss:sumquery}  and~\sectionref{sss:counthistogramquery}, respectively.

Each row of $\Pi[duration]$ and $\Pi[population]$ is a standard aggregate query that points to patients' gender. For instance, the first row of $\Pi[duration]$ adds the values of the number of hospitalization days of \emph{male} patients admitted before June 2022. Now we can batch them by the process of \sectionref{s:batch-coding}, i.e., the $2$-ary polynomial batch codes, and produce:
\begin{align*}
\Pi[duration,population](x) \coloneqq\left\lgroup\!\begin{array}{cccccc}
1   &   1   &   0   &   x   &   0   &   0\\
0   &   0   &   x   &   0   &   x   &   0\\
0   &   0   &   0   &   0   &   0   &   1
\end{array}\!\right\rgroup\in\bigl(\F[x]\bigr)^{3\times6}.
\end{align*}
We can easily check that evaluating $\Pi[duration,population](x)$ component-wise at $x=0$ and $x=1$ recovers $\Pi[duration]$ and $\Pi[population]$, respectively. Each of $\ell$ servers holds an evaluation of $\Pi[duration,population](x)$ at $x=\x$, where $\x[1],\ldots,\x[\ell]\in\F\setminus\{0\mkern0.5mu,\mkern-1mu1\}$ are arbitrary, pairwise distinct scalars. (Note that each such evaluation, a bucket held by a server, is a $2$-batch index of aggregate queries per \definitionref{def:batchindex:aggregate}.) If the client likes to know the count of hospitalized patients who identify as female, she encodes the standard basis vector $\e[1]\in\F^3$ at $x=1$. Moreover, we can enable the client to fetch aggregate information of both of them only in a single round by the batch query mechanism. In that case, she has to encode the same standard basis vector $\e[1]$ at $x=0$ and $x=1$ at once. Thus, we can define the indexes of $k$-batch aggregate queries as follows:

\begin{definition}\label{def:indexbatch:aggregate}
Fix $k>1$ and let $\x[1],\ldots,\x[\ell]\in\F\setminus\{0\mkern0.5mu,\ldots,k-1\}$ be pairwise distinct scalars. A sequence $\oldPi_1,\ldots,\oldPi_{\ell}\in\F^{\pbyr}$ of matrices is an \textsl{index of $k$-batch aggregate queries} for Goldberg's IT-PIR with bucket coordinates $\x[1],\ldots,\x[\ell]$ if, for each $i=1,\ldots,\p$, the matrix
\begin{align*}
    \left\lgroup\ensuremath{\!\!\!\begin{array}{c}\\[-3ex]\pi_{i0}\\[-0.75ex]\vdots\\[-0.25ex]\pi_{i(k-1)}\\[0.25ex]\end{array}\!\!\!}\right\rgroup
\end{align*}
is a simple index of aggregate queries matrix, where, for each $j\in\ival[0]{k-1}$,
\begin{align*}
\pi_{ij}\coloneqq\e\mkern-1.5mu\cdot\mkern-3.5mu\nsum[1.35]_{n=1}^{\ell}\mkern-1.5mu\oldPi_n\cdot\text{\larger[0.5]$\displaystyle\bigl(\mkern-0.5mu\tfrac{j\mkern1mu-\mkern2mux_1}{x_n\mkern1.5mu-\mkern1.5mux_1}\mkern-0.5mu\bigr)\mkern-1.5mu\cdots\mkern-1.5mu\bigl(\mkern-0.5mu\tfrac{j\mkern1mu-\mkern2mux_{n-1}}{x_n\mkern1.5mu-\mkern1.5mux_{n-1}}\mkern-0.5mu\bigr)\mkern-1mu\bigl(\mkern-0.5mu\tfrac{j\mkern1mu-\mkern2mux_{n+1}}{x_n\mkern1.5mu-\mkern1.5mux_{n+1}}\mkern-0.5mu\bigr)\mkern-1.5mu\cdots\mkern-1.5mu\bigl(\mkern-0.5mu\tfrac{j\mkern1mu-\mkern2mux_{\ell}}{x_n\mkern1.5mu-\mkern1.5mux_{\ell}}\mkern-0.5mu\bigr)$}.
\end{align*}
\end{definition}
\vspace*{-0.20cm}
We demonstrate the end-to-end pipeline for our novel IT-PIR pipeline that utilizes the indexes of aggregate queries data structures that we introduce in \figureref{fig:schematic}. Note that the security analysis (described in Appendix~\ref{s:apndx:securityanalysis}) and complexity analysis coincide with that of Hafiz-Henry~\cite{hafiz2017querying}. Thus, with similar $t$-privacy and communication complexity, we achieve a new set of aggregate queries.
\vspace*{-0.25cm}
\section{Evaluation Results}\label{s:eval}
\vspace*{-0.1cm}
\subsection{Benchmarking Experiments}
\vspace*{-0.2cm}

To assess our protocol's query throughput and scaling capabilities, we designed several experiments to benchmark the protocol's performance. These experiments test the first $4$ steps of the proposed algorithm. Benchmarking was performed on a server with $256$ GB RAM and $5$ GPUs. Similar to \cite{hafiz2017querying}, we implement Barrett Reduction \cite{barrett1986implementing} for modulo arithmetic operations. To save storage space, matrices were reduced to compressed column storage representation \cite{duff1989sparse}. For every experiment, the mean value of $100$ trials was reported for each query throughput value. Error bars were negligible for each data point. 
Vector-sparse Matrix multiplication (VspM) throughput on GPU measures how many queries (clients) can be multiplied by an evaluation of the batch indexes of aggregate queries. This measurement includes the time to copy to and from the local host to the GPU device. If a client wants to batch multiple queries together, it does that itself and sends an evaluation of the polynomial-encoded query vector to a server. The server can multiply different query vectors from different clients with its stored evaluation of the batch indexes of aggregate queries, i.e., the GPU throughput of the server. This relationship is linear, and the inverse of the throughput will give the latency per client (query). The server doesn't have to wait for more queries from more clients.

The purpose of the first experiment (as shown in \figureref{fig:benchvaryingr}) is to observe the impact of increasing the number of rows in the database while other variables are maintained at constant values. We vary the number of rows, $\r$, in the database (which is equal to the number of columns in the index of aggregate queries) from $2^{14}$ to $2^{24}$ while keeping the number of rows in the index of queries, $\p$ constant at $2^{16}$ and keeping the number of indexes of aggregate queries batched ($\u$) constant at $6$. Additionally, the number of rows over which aggregation occurs is also fixed at 1. We derive the trend for $128, 256$, and $512$ bit moduli. The general trend observed is an expected monotonic decrease in the throughput with an increasing number of rows in the database. Increasing the bit size of the modulus results in a depreciation of the query throughput attained in most cases. We observe that for a database with over 16 million rows, our protocol is capable of supporting around 4000 queries per second, for modulus with bit sizes of 128, 256 and 512.

\begin{figure}[h]
\includegraphics[width=8.5cm]{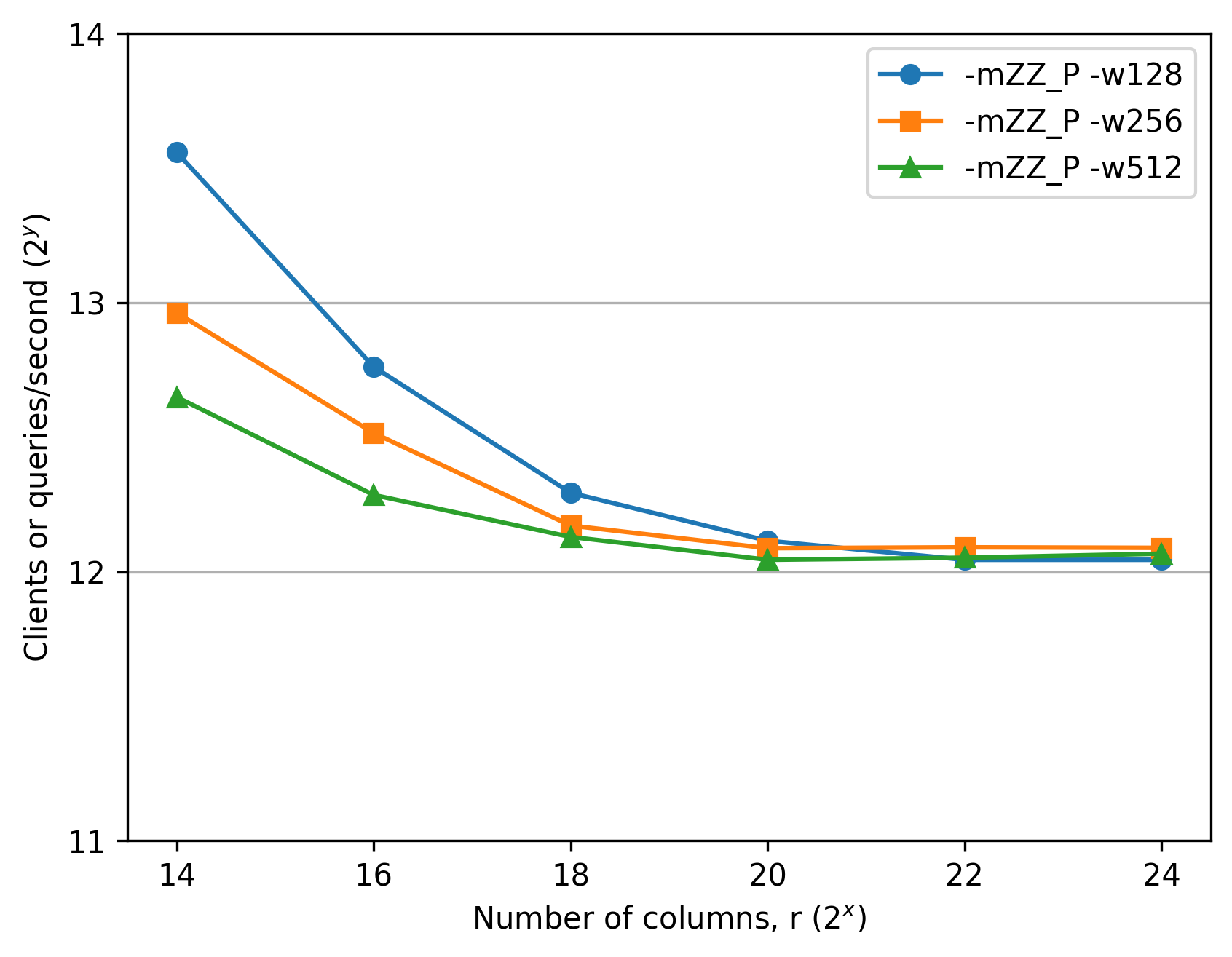}\vspace*{-0.25cm}
\caption{Query throughput with a varying number of rows in the database. The number of rows in the index of aggregate queries, number of rows for aggregation, and number of indexes of aggregate queries batched are kept constant}
\label{fig:benchvaryingr}
\vspace*{-0.65cm}
\end{figure}

The second experiment is very similar to the first, but this time we vary the number of rows in the index of aggregate queries instead of varying the number of rows in the database. We vary the number of rows, $\p$, in the index of aggregate queries from $2^{1}$ to $2^{17}$ while keeping the number of rows in the database, $\r$ constant at $2^{20}$ and keeping the number of indexes of aggregate queries batched ($\u$) constant at $6$. The number of rows over which aggregation occurs is fixed at $2^{4}$. There seems to be no major impact of varying the number of rows in the index of aggregate queries until varied from $2^{16}$ to $2^{17}$, at which point, we observe a drop in the query throughput. Increasing the bit size of the modulus results in a depreciation of the query throughput consistently.

\begin{figure}[h]
\includegraphics[width=8.5cm]{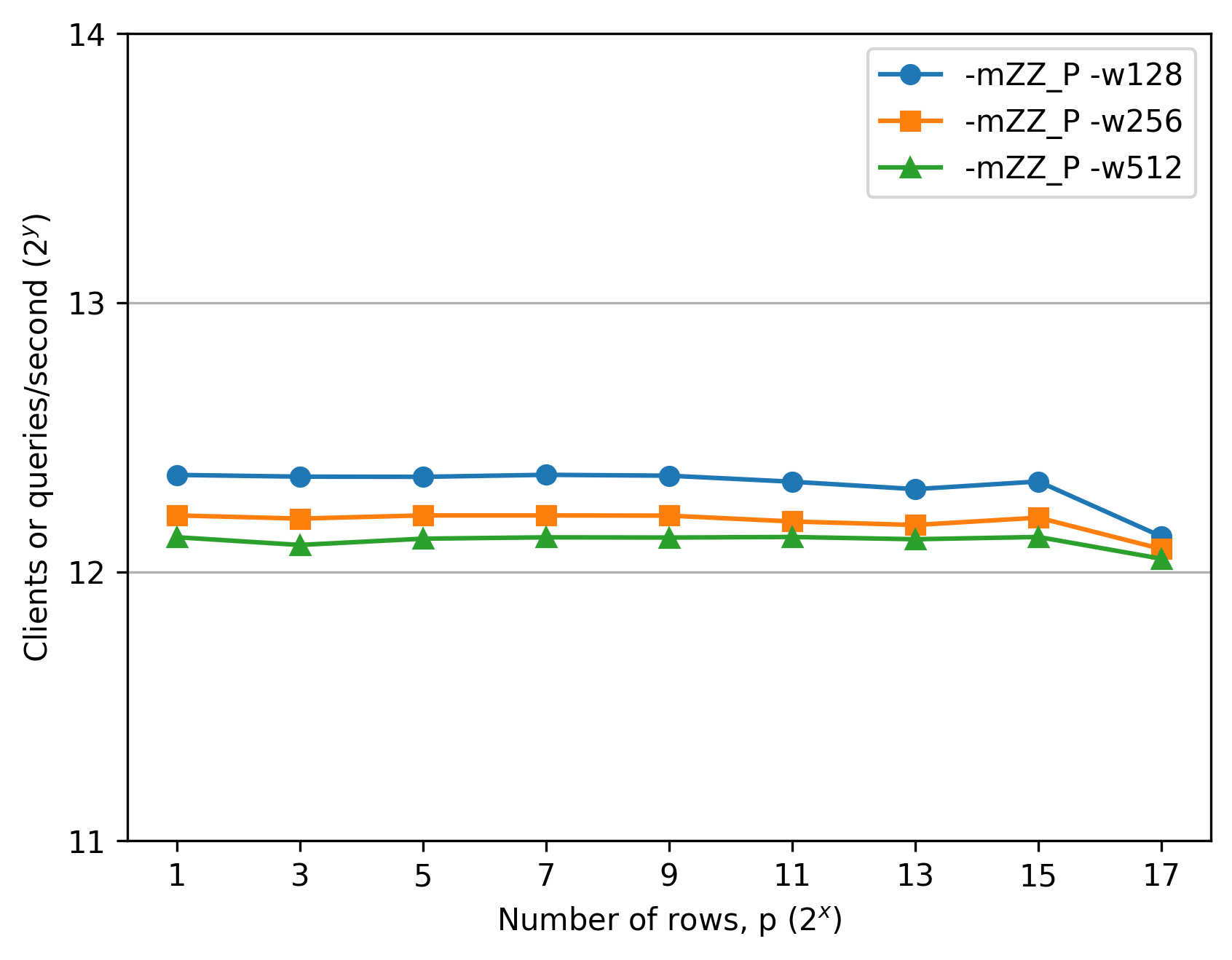}\vspace*{-0.25cm}
\caption{Query throughput with a varying number of rows in the index of aggregate queries. The number of rows in the database, number of rows for aggregation, and number of indexes of aggregate queries batched are kept constant}
\label{fig:benchvaryingp}
\vspace*{-0.57cm}
\end{figure}

The third experiment demonstrates the capability of the protocol to perform aggregations over varying numbers of rows. To evaluate how well the aggregation scales, we maintain the number of rows, columns, and indexes of aggregate queries batched together fixed at $\p=2^{14}$, $\r=2^{16}$, and $\u=6$ respectively. The number of rows over which aggregation occurs varies from $2^{1}$ to $2^{11}$. 
The experiment is performed with a 256-bit modulus.
Other than a single unexpected spike during scaling (when the number of rows over which aggregation occurs is increased to $2^4$ from $2^3$), increasing the number of rows over which aggregation occurs progressively reduces the query throughput until it plateaus out as shown in \figureref{fig:benchvaryingu}. For queries that aggregate over up to 512 rows in the database, the protocol can support around $4{,}000$ queries per second, which is about half the throughput observed for non-aggregate queries on an identical database with a modulus of the same bit size.

\begin{figure}[h]
\includegraphics[width=8.5cm]
{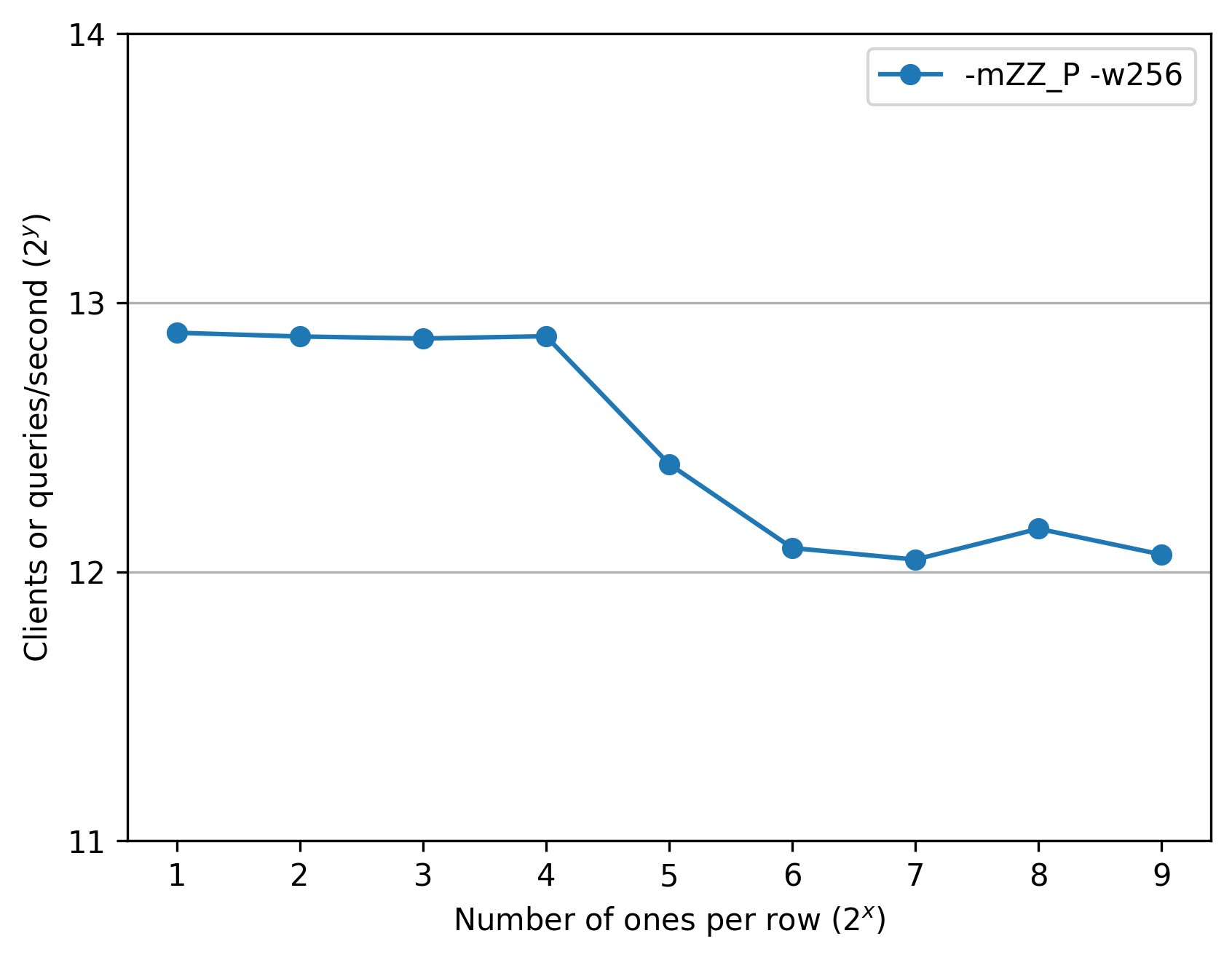}\vspace*{-0.25cm}
\caption{Query throughput with a varying number of rows for aggregation. The number of rows in the index of aggregate queries, the number of rows in the database, the number of indexes of aggregate queries batched, and the bit size of the modulus are kept constant.}
\label{fig:benchvaryingi}
\vspace*{-0.65cm}
\end{figure}

Our final experiment determines the capability of our protocol to scale in terms of the number of queries it can serve by batching several indexes of aggregate queries. For this experiment, the number of rows and columns are fixed at $\p=2^{14}$ and $\r=2^{16}$, respectively, and the number of rows over which aggregation is performed is maintained at $2^4$. The number of indexes being batched, $\u$, is, however, varied from $2^1$ to $2^{10}$. The experiment is performed with a $256$-bit modulus–results are plotted in \figureref{fig:benchvaryingp}. We observe that the time required to perform the batching operation monotonically increases with an increase in the number of files being batched as expected. The time required to batch a thousand files is about $32$ seconds.

\begin{figure}[h]
\includegraphics[width=8.5cm]
{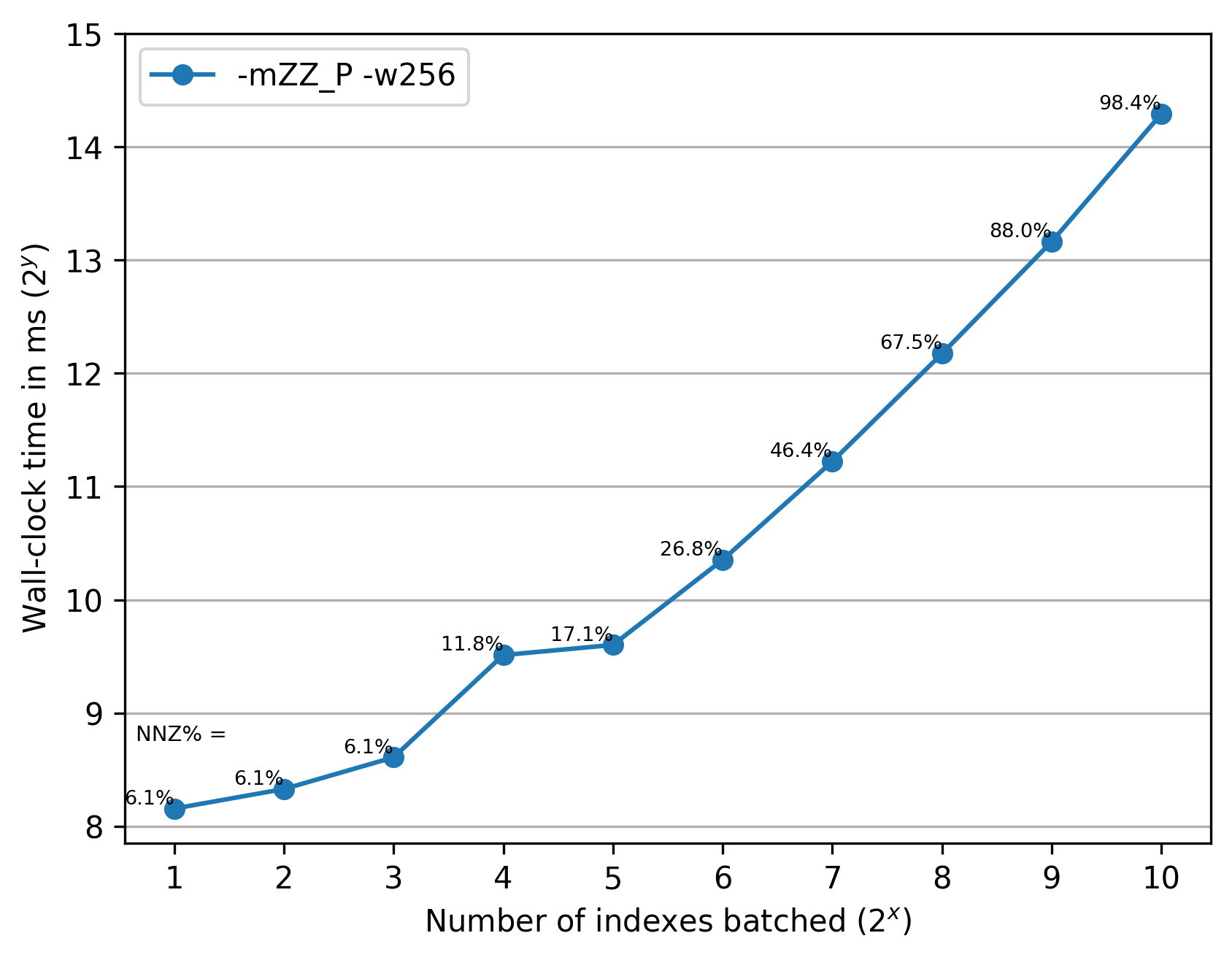}\vspace*{-0.25cm}
\caption{Batching time with a varying number of indexes of aggregate queries batched. The number of rows in the index of aggregate queries, number of rows in the database,  number of rows for aggregation, and bit size of modulus are kept constant. The NNZ percentages refer to the number of non-zero columns in the batched indexes of aggregate queries, and this percentage becomes higher as progressively larger numbers of indexes are batched.}
\label{fig:benchvaryingu}
\vspace*{-0.75cm}
\end{figure}

\subsection{Practical Applications}
\vspace*{-0.15cm}

To demonstrate the practical application of our protocol in real-world situations, we designed three case studies based on real databases. The case studies are primarily concerned with steps $5$ to $7$ of the protocol. The first case study utilizes a medical database, MIMIC3 \cite{johnson2016mimic}, which collects clinical data. Medical datasets are often sensitive and require cautious handling and querying. Additionally, an untrusted database owner who can observe a user's query to a clinical database might be able to infer information about the user based on the information they are looking for in the database. To prevent this, our protocol can be implemented with the medical database. The second case study demonstrates the applicability of our protocol to a gigantic social media database constructed from Twitter. A social media platform query can reveal various information about the user to a malicious database owner who can observe user queries.
Additionally, databases of popular social media platforms tend to host enormous amounts of user data and are thus ideal databases for testing how well our protocol scales and performs on large volumes of data. Our third and final case study compares the performance of our protocol to the most similar instance of our work, a protocol called Splinter, which is also capable of performing private aggregate queries. To perform this comparison, we recreate the same queries that Splinter performs on a synthetic Yelp database, which hosts restaurant reviews.

\vspace*{-0.25cm}
\subsubsection{Case study I: MIMIC3 Medical Database}
Mimic 3 is a freely available clinical database comprising de-identified data of a large number of patients admitted to the Beth Israel Deaconess Medical Center between 2001 and 2012. The database contains data in various formats, but our case study mainly deals with tabular data. Numerous tables within the database outline details about the hospitalization of patients and the details of the care that they received during their stay. We design several practical aggregate queries on some of these data. Specifically, we use the ``Admissions'' and ``Prescription'' data for our queries. 
The queries performed can be broadly divided into two groups, based on which of them were batched together to form indexes of aggregate queries.
The first group of queries is accumulated together into an index of aggregate queries in which every row corresponds to a value in the column that stores the admission type for a particular patient, thus providing an individual's hospitalization's nature or circumstances. These queries are all based on the Admissions table, which has $58{,}976$ rows. For this experiment, $4$ indexes of aggregate queries of dimension $4$x$58{,}976$ were batched.
The queries for which the matrices were constructed for batching into this index of aggregate queries are:

\begin{itemize}
\item Count of urgent admissions: This is a query to compute the number of admissions to the hospital that were categorized as `urgent.' {\fontfamily{cmr}\selectfont\color{gray!75!black}\small SELECT COUNT(*) FROM admissions WHERE admission\_type = ‘URGENT.’}
\normalfont

\item Mean hospital stay duration for patients of Hispanic ethnicity categorized as urgent admissions: This query looks more closely at admissions to the hospital that were categorized as `urgent' and filters out the Hispanic population. It then computes the mean duration of stay for this filtered subset of patients.
{\fontfamily{cmr}\selectfont\color{gray!75!black}\small SELECT MEAN(hospitalization\_duration) FROM admissions WHERE admission\_type = ‘URGENT’ AND ethnicity = ‘HISPANIC.’}
\normalfont 

\item Most recent admission categorized as urgent: This query fetches the urgent hospital admission with the most recent timestamp value, thus performing a MAX aggregation:
{\fontfamily{cmr}\selectfont\color{gray!75!black}\small SELECT MAX(admittime) FROM admissions WHERE admission\_type = ‘URGENT.’}
\normalfont

\item Histogram of admission types: This is a query to obtain the counts of the different buckets in the admission\_type column. So, essentially it is a number of count queries batched together. The equivalent SQL query is:
{\fontfamily{cmr}\selectfont\color{gray!75!black}\small SELECT COUNT(*) FROM admissions WHERE admission\_type = A.}
\normalfont Where A represents a possible value in the admission\_type column.
\end{itemize}

While the previous set of queries was all batched at the level of admission type, we also constructed indexes of aggregate queries that were batched at the patient level. This results in indexes of aggregate queries of dimension $1{,}400$x$4{,}156{,}450$, since each row of the matrix corresponds to one of the $1{,}400$ Hispanic patients in the dataset and the database constructed by joining the admissions and prescription tables has $4156450$ rows. For this experiment, two indexes of aggregate queries were constructed and batched. The queries that these matrices serve are:

\begin{itemize}
\item Total number of doses of drugs administered to a Hispanic patient admitted as an urgent admission: This query aggregates over the dose\_val\_rx field for each patient to obtain the desired response for a particular patient. {\fontfamily{cmr}\selectfont\color{gray!75!black}\small SELECT SUM(dose\_val\_rx) FROM admissions WHERE admission\_type = ‘URGENT’ AND  subject\_id = ‘100012.’} \normalfont

\item Total Hospitalization Duration of Hispanic patient admitted as an emergency admission: This query aggregates over the hospitalization duration field for a patient receiving an emergency admission, thus performing addition over stay durations across every instance of emergency hospitalization of a particular patient present in the database. {\fontfamily{cmr}\selectfont\color{gray!75!black}\small
SELECT SUM(hospitalization\_duration) FROM admissions WHERE subject\_id = ‘100012’ AND admission\_type = ‘EMERGENCY.’}
\normalfont
\end{itemize}

\vspace*{-0.15cm}
\subsubsection{Case Study II: Twitter Database}
Twitter is a very popular social media platform with a sizable active user base and, thus, hosts an enormous amount of data that users can query.  A malicious database owner capable of observing the details of incoming queries to the database can use the information gained from the queries to make inferences about the users querying the database. For instance, numerous users query Twitter about political information. Observing these queries can help a database owner make inferences about a user’s political affinity or, more generally, a user’s political interests, which is potentially sensitive information that a user would normally prefer to keep private. Additionally, tracking a user’s queries over a period of time can allow the database owner to make inferences about a user’s preferences and interests that can be used for targeted advertising. It is thus useful for such databases to incorporate a scheme for privatizing query information. Owing to the volume of data on Twitter, we decided to scrape a large number of records from Twitter to implement our scheme and demonstrate its scalability for application to large real-world databases. We extracted around a million tweets from Twitter that have political relevance using popular hashtags such as 'USElections', 'congress', 'Trump' and more and limited the search space for scraping to a 2 year period around US elections. The scraped database is around $400$ MB in size.

We devise indexes of aggregate queries to serve four distinct queries to the database and batch them together. Each row of the batched index of aggregate queries corresponds to a distinct user in the database, and the number of columns corresponds to the number of records in the database. For this experiment, $2$ indexes of aggregate queries were constructed and batched. The queries are:

\begin{itemize}

\item Total likes received across all tweets made by a particular user: This query aggregates all the likes (positive responses) that any specific individual has obtained and provides a sense of their popularity. {\fontfamily{cmr}\selectfont\color{gray!75!black}\small
SELECT SUM(like\_count) FROM twitter\_data WHERE user\_id = ‘100012.’} \normalfont

\item Count of tweets by a particular user that received no retweets:
Similar to the previous query, the count of tweets by a user that was not shared by any other user also provides a sense of how popular or influential the user is. {\fontfamily{cmr}\selectfont\color{gray!75!black}\small
SELECT COUNT(*) FROM twitter\_data WHERE user\_id = ‘100012’ AND no\_retweets = 0. }\normalfont

\item Mean number of likes on posts made by a politician: The average number of likes on the tweets made by a politician’s account indicates their popularity, and observing a user seek this information also sheds some light on their political interests and affinities. {\fontfamily{cmr}\selectfont\color{gray!75!black}\small SELECT MEAN(like\_count) FROM twitter\_data WHERE user\_id = ‘100012.’ }\normalfont

\item Histogram of likes on the posts of a set of politicians: To do a comparative analysis of the popularity of a group of politicians, a histogram can be computed for the likes on the tweets made by the group of politicians. This is achieved by batching a number of SUM queries, each computing the sum of likes on the tweets of a single politician. {\fontfamily{cmr}\selectfont\color{gray!75!black}\small SELECT SUM(like\_count) FROM twitter\_data WHERE user\_id in (‘100012’, ‘100013’, ‘100014’) GROUP BY user\_id. }\normalfont 
\end{itemize}

\vspace*{-0.25cm}
\subsubsection{Case Study III: Yelp Database}\label{subsec:yelp_cs}
In order to compare our protocol with an existing scheme for performing private aggregate queries called Splinter, we recreate an experiment presented by the authors of Splinter. Splinter uses the Yelp academic dataset to emulate querying a restaurant review website. 
The queries implemented are:
\begin{itemize}
\item 
{\fontfamily{cmr}\selectfont\color{gray!75!black}\small SELECT COUNT(*) WHERE category="Thai."}

\item 
{\fontfamily{cmr}\selectfont\color{gray!75!black}\small SELECT TOP 10 restaurant WHERE category="Mexican"  AND (hex2mi=1 FOR hex2mi=2 OR hex2mi=3) ORDER BY stars.}

\item {\fontfamily{cmr}\selectfont\color{gray!75!black}\small SELECT restaurant, MAX(stars) WHERE category="Mexican" OR category="Chinese" OR category="Indian" OR category="Greek" OR category="Thai" OR category="Japanese" GROUP BY category.}
\normalfont
\end{itemize}

We only demonstrate the first and the third query using our protocol since the second query requires precomputing the hex2mi field that was part of the preprocessing step for Splinter, and recreating it requires explicit details of the exact computation that generates the field. The attribute details for each case study are available in Appendix~\ref{s:apndx:attributedetails}.

 \begin{table*}[!ht]
 \centering
 \begin{tabular}{|c|c|c|p{1cm}|p{1cm}|p{1cm}|p{1cm}|p{1cm}|p{1cm}|p{1cm}|p{1cm}|}
 \hline
 {\multirow{1}{*}{\rotatebox[origin=c]{90}{Case study database}}} & 
 \multirow{1}{*}{\rotatebox[origin=c]{90}{Search space}} &
 \multicolumn{1}{p{2cm}|}{Index of aggregate queries matrix for} & 
 \multicolumn{1}{p{1.5cm}|}{Index matrix size, $p\times r$ } &
 \multicolumn{1}{p{1cm}|}{Index matrix generation time (secs)} &
 \multicolumn{1}{p{1.5cm}|}{Multiple index matrices batching time (mins)} &
 \multicolumn{1}{p{1.3cm}|}{Additional data structure storage Size (MiB)} & 
 \multicolumn{1}{p{1.5cm}|}{VspM throughput on GPU (clients/sec)} &    
 \multicolumn{3}{p{3cm}|}{Server response generation time (secs) per query} \\ 
\cline{9-11}
 &
 &
 &
 &
 &
 &
 &
 &
 All Attr.&
 Essen. Attr. &
 Regular IT-PIR \\
 \hline
 \parbox[t*]{2mm}{\multirow{4}{*}{\rotatebox[origin=c]{90}{Twitter}}} &
 {\multirow{2}{*}{\rotatebox[origin=c]{90}{All}}} &
 Likes &
 \multirow{1}{*}{$333{,}286$$\times$} &
 $1.98$&
 \multirow{2}{*}{$4.450$}&
 \multirow{2}{*}{$22.77$}&
 \multirow{2}{*}{$4{,}927.76$}&
 \multirow{2}{*}{$1.13$}&
 \multirow{2}{*}{$0.45$} &
 \multirow{4}{*}{$1.13$}
 \\
 &
 &
 Retweets&
 $1{,}004{,}129$&
 $7.77$&
 &
 &
 &
 &
 &
  \\
\cline{2-10} &
{\multirow{2}{*}{\rotatebox[origin=c]{90}{Filt.}}} &
 Likes &
 $2{,}714$$\times$&
 $0.06$&
 \multirow{2}{*}{$0.039$}&
 \multirow{2}{*}{$5.30$}&
 \multirow{2}{*}{$5{,}601.52$}&
 \multirow{2}{*}{$0.01$}&
 \multirow{2}{*}{$0.01$} &
 \\
 &
 &
 Retweets&
 $1{,}004{,}129$&
 $0.09$&
 &
 &
 &
 &
 &
 \\
 \hline
 \parbox[t]{2mm}{\multirow{6}{*}{\rotatebox[origin=c]{90}{MIMIC 3}}} &
 {\multirow{4}{*}{\rotatebox[origin=c]{90}{All}}} & 
 Admission Type&
 \multirow{4}{*}{$4 \times 58{,}976$}&
 $0.06$&
 \multirow{4}{*}{$0.002$}&
 \multirow{4}{*}{$0.76$}&
 \multirow{4}{*}{$2{,}0534.12$}&
 \multirow{4}{*}{$0.11$}&
 \multirow{4}{*}{$0.04$}&
 \multirow{4}{*}{$0.11$}
 \\
 &
 &
 Ethnicity &
 &
 $0.39$&
 &
 &
 &
 &
 &
 \\
 &
 &
 Latest Admission&
 &
 $0.98$&
 &
 &
 &
 &
 &
 \\
 & 
 &
 Oldest Admission&
 &
 $0.95$&
 &
 &
 &
 &
 &
 \\
 \cline{2-11}
 & 
 {\multirow{2}{*}{\rotatebox[origin=c]{90}{Filt.}}}
 &
 Dosage 
 & $1{,}400\times$
 &$2.12$
 &\multirow{2}{*}{$0.101$}
 &\multirow{2}{*}{$34.34$}
 &\multirow{2}{*}{$4{,}412.80$}
 &\multirow{2}{*}{$0.26$}
 &\multirow{2}{*}{$0.10$}
 &\multirow{2}{*}{$7.37$}
 \\
 & 
 & 
 Stay Duration 
 &$4{,}156{,}450$
 &$2.03$
 &
 &
 &
 &
 &
 &
 \\
 \hline
 \parbox[t]{2mm}{\multirow{4}{*}{\rotatebox[origin=c]{90}{Yelp}}} &
 {\multirow{2}{*}{\rotatebox[origin=c]{90}{All}}} &
 Restaurant Type 
 &$1{,}312\times$
 &$1.09$
 &\multirow{2}{*}{$0.070$}
 &\multirow{2}{*}{$5.33$}
 &\multirow{2}{*}{$10{,}364.48$}
 &\multirow{2}{*}{$0.13$}
 &\multirow{2}{*}{$0.03$}
 &\multirow{2}{*}{$0.13$}
 \\
 & 
 & 
 Highest Rated 
 &$150{,}346$
 &$10.88$
 &
 &
 &
 &
 &
 &
 \\
  \cline{2-11}
  & 
  {\multirow{2}{*}{\rotatebox[origin=c]{90}{Aug.}}} &
  Restaurant Type &$1{,}312 \times$
 &$3.51$
 &\multirow{2}{*}{$0.061$}
 &\multirow{2}{*}{$21.78$}
 &\multirow{2}{*}{$6{,}585.88$}
 &\multirow{2}{*}{$0.54$}
 &\multirow{2}{*}{$0.14$}
 &\multirow{2}{*}{$0.54$}
 \\
 & 
 & 
 Highest Rated 
 &$601{,}384$
 &$43.06$
 &
 &
 &
 &
 &
 &
 \\
 \hline
 \end{tabular}
 \caption{Experimental results for three case studies. $\p$ is the number of possible search terms and $\r$ is the number of records in the database. The search space column indicates which results correspond to the all-possible search terms (All), and which correspond to filtered versions of the queries with the search spaces reduced (Filt.) for Twitter and MIMIC 3. For Yelp, the Aug. results correspond to the results obtained for the scaled-up version of the dataset. Index matrix generation time measures the time required to scan the database and generate the individual indexes of aggregate queries in CCS representation. Index matrix batching time measures the time required to interpolate and evaluate the polynomials required for batching. Data structure storage size is the size of the CCS representation files of the batched indexes of aggregate queries. VspM throughput measures the number of clients (queries) being served per second on a GPU. The server response generation times measure the time required to generate the response with all the attributes of the database and with only the attributes undergoing aggregation using our protocol, and for Goldberg's regular IT-PIR with a positional query, respectively. For all timing measurements, the mean values of $100$ measurement trials are reported, error bars were negligible and were thus not reported.}
 \vspace*{-0.8cm}
 \label{tab:table_1}
 \end{table*}

The indexes of aggregate queries for the first set of queries on the MIMIC 3 database were $4\times58{,}976$ in dimension and were thus consistently faster to generate than the $333{,}286\times1{,}004{,}129$ dimensional Twitter database indexes of aggregate queries or the other MIMIC 3 indexes of aggregate queries which were of dimension $1{,}400\times4{,}156{,}450$. Overall, the time required for generating indexes is consistently of the order of a few seconds, with the index for aggregating hospitalization duration of any patient taking the longest for generation, at around $30$ seconds. 
It should be noted that generating indices is typically a one-time activity and should ideally be performed only once, when the database itself is updated. The same applies to the process of batching the indexes of aggregate queries, most of which require less than a minute, which is sufficiently fast, given the fact that this would be a single time operation performed everyday in a practical implementation. The storage requirements for the indexes of aggregate queries are also nominal, with the larger MIMIC 3 queries requiring the most space, but they require just around $34$ MB for storage. The server response generation times for each of the queries are remarkable for a database with a million records. The maximum time the protocol takes for a query is around $0.45$ seconds, which is for the Twitter queries.
The performance of the YELP queries were also comparable to those of Twitter and MIMIC 3. To mirror the queries in Splinter \cite{wang2017splinter} exactly, we duplicated the records in the database to scale it by a factor of $4$ like they did, and found the server response time and the storage size to increase linearly. In addition to the initial versions of the Twitter queries, we applied a filter on the users to only include those who had at least a single tweet with over $100$ likes, in order to shrink the search space to a size that may be somewhat more common for a regular query to a database. This significantly improved index generation time, GPU throughput, as well as server response generation time. Additionally, it requires a fraction of storage space when compared to the extreme case with all the users included.

In order to accelerate the server response generation, zero-column operations can be skipped entirely to save compute time, but doing so is contingent on the structure of the indexes of aggregate queries. As in the case of Twitter, if the indexes of aggregate queries are such that there are no columns where every element is $0$, there will be no benefit obtained from this optimization, and this is reflected by the fact that the server response generation throughput both with and without zero column skipping implemented is $1.13$ secs/query for the unfiltered Twitter queries. The server response generation throughput can be further optimized by constraining the servers to only operate on and produce the fields in the database on which aggregations are performed. Since that is a much smaller subset of the all the fields in the database, the server response generation becomes significantly faster for every query.  

\paragraph*{Comparison with Baseline}In order to obtain a baseline to compare our protocol against, we evaluate the queries in our case studies using Goldberg’s IT-PIR technique~\cite{goldberg2007improving}. The performance of Goldberg’s protocol is reported in the 3rd sub-column, "Regular IT-PIR" of the last column, "Server response generation time per query" in Table \ref{tab:table_1}. While our protocol produces a server response from the 4 million record-long MIMIC-3 database in a fraction of a second, Goldberg’s protocol requires over 7 seconds. Unlike our protocol, Goldberg’s technique would require multiple such interactions between client and server in order to serve expressive aggregate queries, thus increasing the actual query response generation time significantly. Goldberg’s protocol is designed to send positional queries to the servers and so, aggregate queries with conditions such as searching or sorting require multiple rounds of interaction with the database, consequently requiring significantly more time in several circumstances to produce results for aggregate queries when compared to our protocol. Goldberg’s protocol can be described as an atomic unit for more complex queries like aggregate queries. The speedup obtained by our protocol over the baseline Goldberg’s protocol is further enhanced when operating with only the necessary attributes for computations, as evidenced by the results presented in the 2nd sub-column, "Essen. Attr." of the last column, "Server response generation time per query" in Table \ref{tab:table_1}. 

The steps common to our protocol and the baseline Goldberg’s IT-PIR protocol are performed on the CPU. This first entails client query generation on CPU and distribution using Shamir’s secret sharing. The secret share vectors sent to the servers undergo an additional step in our protocol that is not a part of the baseline protocol. This step involves a multiplication of the vectors with the sparse IAQ matrices, and is measured on the GPU in order to accelerate the protocol. This produces a vector which is multiplied with the database matrices on CPU on each server, and this step is similar to the step in the baseline protocol where the secret share vector is multiplied with the database matrices. These steps are evaluated on the CPU for both protocols. The same server is used for evaluating both protocols.
\vspace*{-0.15cm}
\section{Discussion}\label{s:apndx:discussion}
\vspace*{-0.15cm}
\subsection{Database Updates}Enterprise databases have a fixed refresh rate (often daily), and thus, the database could potentially change with every refresh. A client query is associated with a specific version of the database, and so a new version of the database would require a new client query, as is common for any PIR protocol. The process of updating the indexes of aggregate query can be automated and is quite trivial owing to them being sparse and being stored in compressed column storage (CCS) representation. This can be achieved by a stored procedure for noting the updated values resulting from the database refresh and scanning the CCS files, and updating the corresponding stored values (both in index as well as pointer files, in case the number of non-zero elements change as a result of the database refresh). In case an index is required to be generated from scratch, a one-time effort (for a specific version of a database) is required to employ our efficient hash map-based algorithm for scanning and constructing the matrices. 

Suppose the client wishes to execute a query that the existing IAQs cannot serve. Combining that IAQ would require a similar effort. This involves a simple scanning operation on the relevant database attribute in order to produce the CCS-represented files and batch them with existing IAQ matrices of the same dimension. The orders of time for these operations can be estimated from some of the examples in TABLE \ref{tab:table_1}.  These updates should be reflected in every server hosting the database and the IAQ matrices. As per the protocol, the length of client queries is made to be equal to the number of rows ($\p$) in the IAQ matrices, not the number of records ($\r$) in the database unless $\p = r$. Thus, the length of client queries is only impacted if and when a  database refresh results in a change in the number of search terms in an attribute of interest, resulting in a change in the number of rows in the corresponding IAQ matrix. This protocol, like other PIR protocols, is designed to obfuscate reading access patterns only, not write patterns. There are further complexities that emerge when integrating novelties into existing live systems, and the scope of this research is restricted to the design and evaluation of a novel protocol, with a preliminary exploration of practical implementations. 

\vspace*{-0.25cm}
\subsection{Minimum Configuration}
\vspace*{-0.15cm}
Being a $\t$-private $\u$-batch IT-PIR protocol, the minimum $\t$ and $\u$ is 1, in which case, the number of servers $\ell$ is 2. For our $\t$-private $\u$-batch IT-PIR scheme, $\ell$ must be at least $\t+\u$ to let the client reconstruct the correct record. Thus, for the minimal case where both $\t$ and $\u$ are both equal to 1 ($\u$=1 would mean that there is only a single index of aggregate queries, and thus, there is no batching), the desired query response can be reconstructed only when both servers produce the correct response, and there is no collusion. However, if additional servers are incorporated, as long as any two (or more) of the servers are not colluding, the query response can be successfully reconstructed. This allows our protocol to tolerate collusion of a number of servers, as long as at least $\t+\u$ servers do not collude.

\vspace*{-0.25cm}
\subsection{Byzantine Robustness}\label{ss:byzantinerobust}
\vspace*{-0.15cm}
Most modern distributed storage and computation schemes rely on servers that are owned by third parties. In light of the current landscape of cyber security and privacy, it is not unreasonable to believe that one or more of the servers in a multi-server scheme can become compromised at some point and either not respond or produce incorrect responses. Thus, it is important to take the Byzantine robustness of multi-server schemes into account when considering their practical viability. Our model can be extended to a $\v$-Byzantine-robust model from our basic assumption of an honest-but-curious (passive) adversary. Devet et al.~\cite{devet2012optimally} proposes $\t$-private $\v$-Byzantine-robust vector-matrix model-based IT-PIR protocols that our protocol conforms to as well owing to its (algebraic) linearity property, allowing it to fit into any vector-matrix model-based IT-PIR protocol. This protocol modifies the client-side configuration of Goldberg’s protocol~\cite{goldberg2007improving} by using improved decoding algorithms to achieve Byzantine robustness equal to the theoretical limit ($\v < k - \t- 1$, where $k$ is the number of servers that respond). Thus, our protocol can tolerate collusion or Byzantine failure of up to $\v$ servers without compromising the integrity of query outputs reconstructed from the server responses. Under this formulation, the adversary is assumed to be passive when $\v = 0$.

\vspace*{-0.15cm}
\subsection{Comparison with Splinter}
\vspace*{-0.15cm}
The major challenge in performing a perfect comparison with Splinter’s~\cite{wang2017splinter} Yelp case study results is that Splinter’s results were measured on a server with 1.9 TB of RAM, which is significantly higher than our 256 GB RAM server. Additionally, while Splinter reports the total response time, a breakdown of the execution times for each step in the process is not provided (with the exception of the information that the network latency of the provider was 14ms). Given these differences in the experimental setup, when comparing the two common queries (Subsection \ref{subsec:yelp_cs}) between the protocols, although Splinter is faster at serving the relatively simpler Restaurant Type query (54 ms vs. our protocol’s 140 ms with essential attributes only), our protocol produces a quicker response time for the more complex Highest Rated query, thus underscoring the practical viability of implementing our protocol for enterprise databases. The code base for the Function Secret Sharing implemented in Splinter is available on Github, but dedicated end-to-end code for reproducing the results reported in the paper and instructions for the same are not publicly available, thus making it challenging to execute both protocols on the same server.

\vspace*{-0.15cm}
\section{Additional Experiments}\label{s:apndx:experiments}
\vspace*{-0.05cm}
\subsection{Server Response Generation For Larger Databases}
\vspace*{-0.15cm}
\begin{table*}[!ht]
    \centering
    \begin{tabular}{|p{1.2cm}|l|l|l|l|l|l|l|l|}
    \hline
        {DB Size} & {Case Study} & {Records} & {Record Size} & {Protocol} & \textbf{GF($2^8$)} & \textbf{GF($2^{16}$)} & {\small$-m\mathbb{ZZ}_P-w128$} & {\small$-m\mathbb{ZZ}_P -w256$} \\ \hline
\parbox[t*]{2mm}{\multirow{4}{*}{$40$ GiB}} &
 {\multirow{2}{*}{\rotatebox[origin=c]{0}{Twitter Filt.}}} &
 \multirow{2}{*}{$1{,}004{,}129$} &
 \multirow{2}{*}{$41.8$ KiB} &
 Baseline &
 \multirow{1}{*}{$16.8$} &
 \multirow{1}{*}{$32.1$} &
 \multirow{1}{*}{$1006.1$} &
 \multirow{1}{*}{$580.6$}
 \\
 \cline{5-9}
 &
 &
 &
 &
 This work &
 \multirow{1}{*}{$0.1$} &
 \multirow{1}{*}{$0.2$} &
 \multirow{1}{*}{$7.7$} &
 \multirow{1}{*}{$4.9$}
  \\
\cline{2-9} &
{\multirow{2}{*}{\rotatebox[origin=c]{0}{MIMIC 3 Filt.}}} &
 \multirow{2}{*}{$4{,}156{,}450$} &
 \multirow{2}{*}{$10.5$ KiB} &
 Baseline &
 \multirow{1}{*}{$17.5$} &
 \multirow{1}{*}{$33.6$} &
 \multirow{1}{*}{$1010.6$} &
 \multirow{1}{*}{$589.3$}
 \\
 \cline{5-9}
 &
 &
 &
 &
 This work &
 \multirow{1}{*}{$0.6$} &
 \multirow{1}{*}{$1.1$} &
 \multirow{1}{*}{$36.6$} &
 \multirow{1}{*}{$21.2$}
  \\  
  \hline

\parbox[t*]{2mm}{\multirow{4}{*}{$64$ GiB}} &
 {\multirow{2}{*}{\rotatebox[origin=c]{0}{Twitter Filt.}}} &
 \multirow{2}{*}{$1{,}004{,}129$} &
 \multirow{2}{*}{$67.0$ KiB} &
 Baseline &
 \multirow{1}{*}{$27.5$} &
 \multirow{1}{*}{$51.0$} &
 \multirow{1}{*}{$1754.1$} &
 \multirow{1}{*}{$988.8$}
 \\
 \cline{5-9}
 &
 &
 &
 &
 This work &
 \multirow{1}{*}{$0.2$} &
 \multirow{1}{*}{$0.3$} &
 \multirow{1}{*}{$12.2$} &
 \multirow{1}{*}{$7.6$}
  \\
\cline{2-9} &
{\multirow{2}{*}{\rotatebox[origin=c]{0}{MIMIC 3 Filt.}}} &
 \multirow{2}{*}{$4{,}156{,}450$} &
 \multirow{2}{*}{$16.5$ KiB} &
 Baseline &
 \multirow{1}{*}{$27.8$} &
 \multirow{1}{*}{$51.7$} &
 \multirow{1}{*}{$1703.8$} &
 \multirow{1}{*}{$981.3$}
 \\
 \cline{5-9}
 &
 &
 &
 &
 This work &
 \multirow{1}{*}{$0.9$} &
 \multirow{1}{*}{$1.5$} &
 \multirow{1}{*}{$49.9$} &
 \multirow{1}{*}{$33.2$}
  \\  
  \hline

  \end{tabular}
  \vspace*{-0.15cm}
  \caption{Server response generation times for different modulus sizes using our IAQ protocol and the baseline protocol on augmented $40$GiB and $64$GiB databases. Reported server response times are in seconds and are the means of $10$ trials.}\vspace*{-0.55cm}
 \label{tab:table_2}
\end{table*}

For our initial experimentation with case studies, the primary objective was to showcase that our protocol can be used for executing realistic queries on real-world databases with a relatively high number of records. However, despite some of the case studies having databases with records in the order of millions, the overall size of the databases was not too big. As a result, we picked two of our case studies (Twitter Filt. and MIMIC 3 Filt. of Table \ref{tab:table_1}) and augmented each record in the corresponding databases to make the total size of the databases $40$ GiB (while keeping the number of records in the corresponding databases the same like before). The Twitter and the MIMIC 3 datasets were (row-wise) augmented $535 \times$
and  $84 \times$, respectively, compared to the original datasets. With these $40$ GiB databases, we executed the same queries as in the main case study for the corresponding databases, and we measured the wall clock time for execution of the queries using our IAQ protocol as well as the baseline Goldberg’s IT-PIR protocol. We also measure our results for a number of different modulus sizes, including GF($2^8$), similar to Hafiz-Henry~\cite{hafiz2017querying}. Since in the binary field, GF($2^8$) operations, each word is 8 bits, and a lookup table is used, multiplications are performed faster than regular arithmetic. As a result, it is a practical and common choice for the modulus. 


As per our observations from Table~\ref{tab:table_2}, GF($2^8$) and GF($2^{16}$) moduli provide faster query response times than the $128$-bit and the $256$-bit moduli. As we increase the modulus (a.k.a. each word size), since the size of each row is kept fixed, the number of words is reduced, and beyond a certain point, increasing the modulus size results in faster query response owing to a reduction in the number of multiplications and additions resulting from the reduction in number of words. Using the GF($2^8$) modulus, our protocol obtains a \textasciitilde$141\times $ improvement in server response generation time over the baseline protocol for the Twitter queries and a \textasciitilde$31\times $ improvement for the MIMIC 3 queries. The server response generation for both Twitter and MIMIC 3 queries closely mirror each other for the baseline protocol since the size of both the databases have been augmented to the same value of $40$ GiB, but since our protocol exploits the sparsity of our derived data structures, which are different for different databases, the server response generation times for the two databases using our protocol does not always shadow each other despite the size of the database being identical.

We also perform the same set of experiments by augmenting the databases to a size of $64$ GiB, and the trends are similar to those observed for experiments on the $40$ GiB databases.


\vspace*{-0.15cm}
\subsection{Benchmarking Query Throughput on CPU}
\vspace*{-0.15cm}
In this section, we discuss benchmarking the query throughput of our protocol on the CPU. In our main results, we perform the step of multiplying client vectors with batch indexes of aggregate queries on the GPU. This was motivated by the fact that most enterprise compute servers that deal with large volumes of data typically have GPU(s) on them. 
In order to assess the performance of the query throughput on the CPU, we replicate the setup of our first benchmarking experiment. For instance, the number of rows ($\p$) in the indexes of aggregate queries is set to $2^{16}$, the number of rows over which aggregation occurs is set to $4$, and the number of indexes batched is kept constant at $6$. The number of columns in the indexes of aggregate queries ($\r$) is varied from $2^{16}$ to $2^{20}$ and the number of clients served per second is observed with various modulus bit sizes, e.g., $2^7$, $2^8$, and $2^9$. Note that our NTL-based CPU implementation of vector-sparse-matrix (VspM) multiplication is very basic and doesn't leverage any parallel processing or hand-optimized assembly coding, contrary to our GPU implementation. There are optimization and engineering tricks available even for CPU implementation that were out of the scope of our attention. 

We observe that the query throughput drops drastically on the CPU, as expected. While with a $2^7$-bit modulus, it is possible to serve $8$ queries per second for a database with $2^{16}$ records\emdash this number diminishes to around $5$ queries per second when the number of records is increased to $2^{20}$ (around a million) and the modulus size is $2^9(=512)$-bit. However, despite a seemingly low throughput, it should be noted that server response generation is performed on the CPU, and our case studies demonstrate that given the order of server response generation times reported in Table \ref{tab:table_1} and Table \ref{tab:table_2}, it is unlikely that this step will become a bottleneck even when computed on the CPU with no sophisticated optimization tricks.


\vspace*{-0.25cm}
\section{Related Work}\label{s:related-work}
\vspace*{-0.15cm}
Despite a relative dearth of literature on expressive querying mechanisms that can support a plethora of aggregate queries, there have been a number of attempts to achieve such protocols due to the importance of aggregate queries in data analytics. Wang et al.  \cite{wang2017splinter} proposes an approach to privatizing aggregate queries called Splinter, which aims to hide sensitive parameters in a user's query by extending function secret sharing \cite{boyle2015function} to hide aggregate query results and specific inputs. Splinter splits a query into multiple parts and sends each part to a different provider, which holds a copy of the data. This query obfuscates only the particular item value that the user desires, such as a specific item ID. The provider then inputs all the items in the shared function it got from the user and sends them back to the user. The user then receives all the results and combines all the shares to produce the desired result. This approach is used in different ways depending on the aggregate query. Splinter, however, has a strict assumption of non-collusion, while our protocol can tolerate collusion of servers up to a configurable threshold number. Additionally, our protocol can perform several operations, such as MIN, MAX, and Top-K, with a single interaction. In contrast, Splinter requires multiple round trips for instances of these queries with a large number of disjoint conditions.

Similarly, recent work attempts to provide users with an interface for private aggregate queries \cite{zhao2022information}. Still, it formulates the construction for a single server K user problem. They present a two-pass model for secure aggregation, where over the first round, any number of users less than U out of K users respond to the server that wants to learn the sum of the users, and the remaining users are viewed as dropped. Over the second round, any number of users less than U out of the surviving users respond. From the information obtained from the surviving users over the two rounds, the server can decode the desired sum. This protocol is not formulated as a means of securely performing an aggregate query for a user – rather, it tries to solve the problem statement of securely computing the sum of a number of inputs provided by a group of users, while some users may not respond. It is thus also limited in the aggregations it can support.

Ahmad et al. \cite{ahmad2021coeus} proposes a scheme somewhat similar to what our protocol aims to achieve. They present a system for oblivious document ranking and retrieval with a pipeline for sending keywords to a server, a ranked list of documents sent back by the server in response, and one of the documents selected and opened. Although sharing similarities with our protocol in that both schemes employ the preparation of specialized data structures for enhancing the process of private information retrieval and are capable of supporting Top-K queries, they are solving inherently different problems because COEUS was explicitly developed for search relevance queries. Thus, Top-K is the only aggregate query it can support. While COEUS focuses primarily on oblivious ranking and retrieval of search results, our protocol primarily focuses on oblivious retrieval of aggregated results from a database.

A number of other recent works attempt to tackle the challenges of securely performing aggregations. For instance, Kadhe et al.  \cite{kadhe2020fastsecagg} proposes a solution scheme based on a finite-field version of the Fast Fourier Transform that also relies on Shamir's secret sharing, but the correctness of the response may not be assured. Additionally, it is designed only for aggregating updates shared by clients to a central server, whereas our protocol is not restricted to aggregations – it can support expressive keyword searches as well. Similarly, Bonawitz et al.  \cite{bonawitz2017practical} proposes a multi-party computation-based solution for secure aggregation that also uses Shamir's secret sharing. Pillutla et al. \cite{pillutla2022robust} and So et al.  \cite{so2020byzantine} attempt to solve the secure aggregation problem specifically in the setting of federated learning. However, the problem that these works attempt to solve involves collecting updates from numerous clients and aggregating them in an incorruptible server, and they attempt to solve certain problems that are not necessarily a concern for our protocol, such as malicious clients attempting to corrupt the central aggregation by sharing poisoned updates.

Somewhat similar to our protocol, searchable symmetric encryption (SSE) schemes can be used for keyword searches over encrypted documents. To achieve efficiency, SSE schemes are inflicted with some degree of leakage. The vast majority of the literature on SSE only considers the leakage from the encrypted search index, which is only a part of the overall SSE system, and as a result, they do not have as much protection for queries from a system-wide viewpoint \cite{gui2021rethinking}. Additionally, most SSE schemes focus on index retrieval only, and few perform both index retrieval and document retrieval. Our protocol has a major benefit over SSE schemes because it can privately search and retrieve documents from a database without requiring any encryption for the database and can also retrieve the actual documents, unlike the majority of SSE schemes that only perform index retrieval. Although there exist SSE schemes capable of supporting conjunctive and disjunctive queries \cite{bag2023two}, they aren't typically capable of performing private aggregate queries and, thus, can, at best, support a subset of our protocol's capabilities.

Oblivious RAM (ORAM) attempts to solve a similar problem as PIR. While the amortized communication complexity of ORAMs is historically low, and they feature no computation on the server, the lower bounds often make them infeasible to implement for practical applications. Additionally, the feasibility of implementation for ORAMs is further impacted by them often requiring a download and reshuffle of the entire database~\cite{mayberry2013efficient}. There are works that have pushed ORAMs towards sublinear worst-case complexity~\cite{kushilevitz2012security, shi2011oblivious}, but schemes proposing further improvements to communication complexity have come at the cost of increasing client memory from constant to logarithmic~\cite{stefanov2018path} or polynomial~\cite{stefanov2013oblivistore, stefanov2011towards}. For dynamic searchable encryption, ORAM can achieve a worst-case search complexity of $\ O(Mlog(N))$ where$\ N$ is the total number of keyword-file identifier pairs and$\ M$ is the maximum number of matched file identifiers in the index~\cite{wu2023obi}. However, ORAMs can be used in conjunction with PIR schemes to exploit PIRs’ better worst-case guarantees and become feasible for implementation for practical applications~\cite{mayberry2013efficient}. 

Like ORAMs, oblivious datastores also attempt to mitigate similar threats as PIR does. Pancake~\cite{grubbs2020pancake} formulates a passive adversary capable of continuously monitoring accesses to encrypted data. While the assumption that adversaries cannot inject queries allows it to outperform well-known ORAM schemes such as PathORAM~\cite{stefanov2018path}, Pancake requires accurate knowledge of the original access distribution to guarantee security. Waffle~\cite{maiyya2023waffle}, another oblivious datastore, improves upon Pancake by removing the requirement for prior knowledge of access sequences, and improves performance by compromising on security guarantees, since providing completely uniform accesses on the server is expensive. Our protocol does not require prior access patterns and does not compromise on security guarantees for improved performance.

\section{Conclusion and Future Work}\label{s:conclusion}
We presented a novel framework that augments regular IT-PIR protocols (e.g., Goldberg's IT-PIR) with aggregate queries, e.g., SUM, COUNT, MEAN, Histogram, etc. We proposed constructions of effective indexes of aggregate queries comprising new standard aggregate vectors. When we conflux these auxiliary data structures with the polynomial batch coding technique, we open up new opportunities, i.e.,  enabling aggregation with searching, sorting, ranking, and various constraints.
We evaluated them by benchmarking and simulating several privacy-preserving real-world applications that enable a user concerned about privacy to submit aggregate queries to untrustful databases obliviously. Our results show that efficient implementation of our framework on GPU can achieve fast query response time while assuring the privacy of aggregate queries. An interesting avenue for future research would be to explore techniques that support more complex queries, e.g., various JOIN operations, potentially utilizing different batch coding and indexing schemes.

\section*{Acknowledgement}
We thank Ryan Henry for collaborating on the prior work~\cite{hafiz2017querying} and for supervising the PhD dissertation research~\cite{hafiz2021privateir} of the leading author, Syed Mahbub Hafiz, upon which the ideas developed in this paper are built. We are thankful to NDSS anonymous reviewers of this paper and the associated artifact, including our anonymous shepherd, for their valuable comments to improve the manuscript and artifact. This project is based upon work supported partly by the UC Noyce Institute: Center for Cybersecurity and Cyberintegrity (CCUBE).
\bibliographystyle{ieeetr}
\bibliography{refs}
\appendices

\section{Attribute Details and Additional Information for Case Studies}\label{s:apndx:attributedetails}
Through our case studies, we attempt to demonstrate the scalability of our protocol with respect to the number of records in the database being queried. Each of our case studies has a different number of rows in the database, with the MIMIC-3 database having 4 million records. Additionally, we further demonstrate the scalability of our protocol through our benchmarking experiments, where we vary the number of records in the database up to over 16 million. Our benchmarking experiments suggest that the impact of a variation in the width (number of columns) of the database is largely inconsequential across the range of variation that was tested, motivating us to focus more on the variation in the number of rows in the database.
\vspace*{-0.15cm}
\subsection{MIMIC 3}
\vspace*{-0.15cm}
The fields required for constructing the indexes of aggregate queries, performing filtering, and observing aggregations are the type of admissions, ethnicity of patients, patient id, admission and discharge time of patients, and drug dosage. 
To support some of the queries, we perform some preprocessing on the database to derive a few columns from existing columns in the MIMIC 3 database. These involve computing a duration of hospitalization field using the admit time and discharge time fields, an additional column that demarcates whether the ethnicity of any given patient is Hispanic or not, and finally, an integer id column for type of admission. Overall, there are 9 fields, $4$ Bytes each, resulting in each record being $36$ Bytes in size. However, in order to serve all the queries in this case study, aggregation occurs only over $4$ of those fields (admission type id, field indicating whether a patient is Hispanic or not, hospital duration and drug dosage), thus providing the option of retaining only those fields when observing PIR throughput for the queries.
\vspace*{-0.15cm}
\subsection{Twitter}
\vspace*{-0.15cm}
Publicly available databases are often scraped for usage in research, and Hafiz-Henry~\cite{hafiz2017querying} is an example of this, where they scrape the publicly available IACR-ePrint archive. There are other works~\cite{guan2022census, kusumasari2020scraping, al2019data, singh2016comparison} that scrape publicly available Tweets in order to use them for their research. Besides, Twitter/X provides a Python library (https://www.tweepy.org/) to do similar tasks, which suggests that they are ethically not opposed to the usage of publicly shared tweets. Scraped tweets have been used for a broad spectrum of research, as exemplified by some of these works, which cover social network analysis for social sciences, identification of crucial needs and responses during disasters, linking patients with common illnesses to form enriched datasets for future research and comparison of the performance of classifiers.

Among the scraped fields, the ones that we use for building the indexes of aggregate queries and for aggregating are the number of likes on the tweets, the number of reshares on each tweet, the raw tweet itself, and also the user ids associated with each tweet.
Similar to the MIMIC 3 database, we perform some preprocessing on the database to generate some additional columns to aid with counting. The first additional column is an integer user\_id column with distinct integer values for every distinct user in the database, and the other additional column is a binary column to indicate whether a tweet was retweeted at least once or not. We use $5$ attributes per record, with each record being $276$ Bytes in size. However, since aggregation occurs only over $3$ of those fields (integer user id field, count of likes, and the field that indicates whether a tweet has been reshared or not), we are required only to retain those fields (total $12$ Bytes) to observe results of aggregations. 
\vspace*{-0.35cm}
\subsection{Yelp}
\vspace*{-0.15cm}
We utilize the fields business id, stars, and category. We perform preprocessing to add a field that indicates whether a restaurant belongs to the category `Thai.' This results in each record possessing $4$ attributes, each of size $4$ Bytes, leading to the size of each record being 16 Bytes.

\section{Security Analysis in the light of Hafiz-Henry~\cite{hafiz2017querying}}\label{s:apndx:securityanalysis}
The difference between our protocol and Goldberg's IT-PIR protocol is the inclusion of indexes of aggregate queries. We prove that even after the inclusion of both a single index of aggregate queries and $\u$-batch index of aggregate queries, the protocol is still $\t$-private.
\vspace*{-0.15cm}
\subsection{Single index of aggregate queries is $\t$-private}
\vspace*{-0.15cm}
We claim that a $\t$-private Goldberg's IT-PIR query through a single (unbatched) index of aggregate queries $\oldPi$ is still $\t$-private. 
The definition of regular $t$-privacy requires, for every coalition $\coalition\subseteq\ival{\ell}$ of at most $\t$ servers and for every record index $i\in\ival{\r}$, that
\begin{align}
    \Pr\bigl[I=i\bigm\vert Q_{\coalition}=(\oldPi\semicolon\vecQ[\j_1],\ldots,\vecQ[\j_{\t}])\bigr]&=\Pr\bigl[I=i\bigr],\label{eq:equal}
\end{align}
where the random variables $I$ and $Q_{\coalition}$ denote respectively describing the block index the user requests and the joint distribution of share vectors it sends to servers in $\coalition$ (including the ``hint'' that the query would go through $\oldPi$).

However, in case some of the blocks of $\D$ cannot be accessed through $\oldPi$, \equationref{eq:equal} requires modification. Multiple different indexes of aggregate queries are likely to exist, each having an associated \emph{conditional distribution} for $I$, and it is not adequate to constraint $I$ to span solely across blocks accessible via $\oldPi$. A sound definition would require incorporating the notion of curious PIR servers updating their priors by utilizing the information they gain upon observing which specific index of aggregate queries a user request passes through. In order to reflect this idea, we modify the $\t$-privacy definition as follows.
\begin{definition}\label{appndx:def:tprivacy}
Let $\D\in\F^{\rbys}$ and let each $\oldPi_1,\ldots,\oldPi_n$ be an index of aggregate queries for $\D$. Requests are \textsl{$\t$-private with respect to $\oldPi_1,\ldots,\oldPi_n$} if, for every coalition $\coalition\subseteq\ival{\ell}$ of at most $\t$ servers, for every set of record indexes $\setofindexes\in\ival{\r}^{\numberofones}$, and for every index of aggregate queries $\oldPi\in\{\oldPi_1,\ldots,\oldPi_n\}$,
\begin{align*}
    \Pr\bigl[I=\setofindexes\bigm\vert Q_{\coalition}=(\oldPi\semicolon\vecQ[\j_1],\ldots,\vecQ[\j_{\t}])\bigr]&=\Pr\bigl[I=\setofindexes\bigm\vert E_{\oldPi}\bigr],\label{eq:equal-modified}
\end{align*}
where $I$ and $Q_{\coalition}$ denote the random variables respectively describing the record indexes the user requests (to aggregate) and the joint distribution of query vectors it sends to servers in $\coalition$ (including the ``hint'' that the query would go through $\oldPi$), and where $E_{\oldPi}$ is the event that the request passes through $\oldPi$. (Recall the set $\setofindexes$ and its cardinality $\numberofones$ notations from~\definitionref{def:aggregatebasis}.)
\end{definition}
\vspace*{-0.3cm}
Notice that the privacy guarantee is identical between a $\t$-private expressive query through a basic index of aggregate queries $\oldPi$ and a $\t$-private aggregate query over the database $\D_{\oldPi}\coloneqq\oldPi\cdot\D$. This results in \definitionref{appndx:def:tprivacy} reducing to a usual $\t$-privacy definition when $\Pi\in\F^{\rbyr}$ is the identity matrix. 

This observation, along with the $\t$-privacy of Goldberg's IT-PIR~\cite{goldberg2007improving,henry2016polynomial}, leads us to the next theorem.
\newcommand*{\thmsimple}{Let $\D\in\F^{\rbys}$ and let each $\oldPi_1,\ldots,\oldPi_n$ be a simple index of aggregate queries for $\D$. If $\oldPi\in\{\oldPi_1,\ldots,\oldPi_n\}$ with $\oldPi\in\F^{\pbyr}$ and if \smash{$(\x[1],\vecQ[1]),\ldots,(\x[\ell],\vecQ[\ell])$} is a component-wise $(\t+1,\ell)\mkern2mu$-threshold sharing of a standard basis vector $\e[]\in\F^{\p}$ from a client, then $\ell$ tuples, \smash{$(\oldPi,\x[1],\vecQ[1]),\ldots,(\oldPi,\x[\ell],\vecQ[\ell])$}, are $\t$-private with respect to $\oldPi_1,\ldots,\oldPi_n$.}
\begin{theorem}\label{THM:SIMPLE}
\thmsimple
\end{theorem}
\vspace*{-0.35cm}
\begin{proof}
We consider a coalition $\coalition$ comprising $\t$ servers and fix $\setofindexes\in\ival{\r}^{\numberofones}$ and $\oldPi\in\{\oldPi_1,\ldots,\oldPi_n\}$ with $\oldPi\in\F^{\pbyr}$. $I$ denotes the random variables describing the indexes (within $\D$) of the blocks (to aggregate and) requested by the user, $J$ denotes the index of the standard basis vector encoded by the user in the query, and $Q_{\coalition}$ corresponds to the joint distribution of share vectors, along with the ``hint'' indicating that the query passes through the index of aggregate queries matrix $\oldPi$, sent to the servers in $\coalition$.

As per \definitionref{appndx:def:tprivacy}, we need to show that $$\Pr[I=\setofindexes\mid Q_{\coalition}=(\oldPi\semicolon\vecQ[\j_1],\ldots,\vecQ[\j_{\t}])]=\Pr[I=\setofindexes\mid E_{\oldPi}],$$ where $E_{\oldPi}$ denotes the event that the user's request is via $\oldPi$. The key observation underlying the proof is that $$\Pr[I=\setofindexes\mid E_{\oldPi}]=\Pr[\e[J]\cdot\oldPi=\e[\setofindexes]\mid E_{\oldPi}]$$ 
and $\Pr[I=\setofindexes\mid Q_{\coalition}=(\oldPi\semicolon\vecQ[\j_1],\ldots,\vecQ[\j_{\t}])]=\Pr[\e[J]\cdot\oldPi=\e[\setofindexes]\mid Q_{\coalition}=(\oldPi\semicolon\vecQ[\j_1],\ldots,\vecQ[\j_{\t}])].$

Therefore, we have
\begin{align*}
	\Pr\bigl[I=\setofindexes\bigm\vert E_{\oldPi}\bigr]
		&=\Pr\bigl[\e[J]\cdot\oldPi=\e[\setofindexes]\mid E_{\oldPi}\bigr]\\
		&=\Pr\bigl[\e[J]\cdot\oldPi=\e[\setofindexes]\bigm\vert Q_{\coalition}=(\oldPi\semicolon\vecQ[\j_1],\ldots,\vecQ[\j_{\t}])\bigr]\\
		&=\Pr\bigl[I=\setofindexes\bigm\vert Q_{\coalition}=(\oldPi\semicolon\vecQ[\j_1],\ldots,\vecQ[\j_{\t}])\bigr].
\end{align*}
Here, the second line utilizes the $\t$-privacy of secret shares, \smash{$(\x[1],\vecQ[1]),\ldots,(\x[\ell],\vecQ[\ell])$}, of the query vector received by $\ell$ servers, individually.
\end{proof}

\vspace*{-0.15cm}
\subsection{$\u$-batch index of aggregate queries is $\t$-private}
\vspace*{-0.15cm}
Here, we prove~\theoremref{THM:UARY:aggregate}.
\newcommand*{\thmuary}{Fix $\u>1$ and $j\in\ival[0]{\u-1}$, and let $\oldPi=\bigl(\oldPi_1,\ldots,\oldPi_{\ell}\bigr)\in\bigl(\F^{\pbyr}\bigr){}^{\ell}$ be $\ell$ buckets of a $\u$-batch index of aggregate queries with bucket (server) coordinates $\x[1],\ldots,\x[\ell]\in\F\setminus\{0,\ldots,\u-1\}$. If \smash{$(\x[1],\vecQ[\j_1]),\ldots,(\x[\ell],\vecQ[\j_{\ell}])$} is a sequence of component-wise $(\t+1,\ell)\mkern2mu$-threshold shares of a standard basis vector $\e[]\in\F^{\p}$ encoded at $x=j$, then \smash{$(\oldPi,\x[1],\vecQ[\j_1]),\ldots,(\oldPi,\x[\ell],\vecQ[\j_{\ell}])$} is $\t$-private with respect to $\oldPi$.}
\begin{theorem}\label{THM:UARY}
\thmuary
\end{theorem}
\begin{proof}
The proof for this theorem is similar to that of \theoremref{THM:SIMPLE}. We consider a coalition $\coalition$ comprising $\t$ servers and fix $\setofindexes\in\ival{\r}^{\numberofones}$ and $\oldPi\in\{\oldPi_1,\ldots,\oldPi_n\}$ with $\oldPi\in\F^{\pbyr}$. $I$ denotes the random variables describing the indexes (within $\D$) of the blocks requested by the user, $J$ denotes the index of the standard basis vector encoded by the user in the query, $K$ denotes the $x$-coordinate at which the standard basis vector is encoded, and $Q_{\coalition}$ corresponds to the joint distribution of share vectors, along with the ``hint'' indicating that the query passes through the $\u$-batch index of aggregate queries $\oldPi$, sent to the servers in $\coalition$.

As per \definitionref{appndx:def:tprivacy}, we need to show that $$\Pr[I=\setofindexes\mid Q_{\coalition}=(\oldPi\semicolon\vecQ[\j_1],\ldots,\vecQ[\j_{\t}])]=\Pr[I=\setofindexes\mid E_{\oldPi}],$$ where $E_{\oldPi}$ denotes the event that the user's request is through $\oldPi$. The key observation is that $$\Pr[I=\setofindexes\mid E_{\oldPi}]=\nsum_{k=0}^{\u-1}\Pr[\e[J]\cdot\pi_k=\e[\setofindexes]\mid K=k,E_{\oldPi}]\cdot\Pr[K=k\mid E_{\oldPi}]$$ and $\Pr[I=\setofindexes\mid Q_{\coalition}=(\oldPi\semicolon\vecQ[\j_1],\ldots,\vecQ[\j_{\t}])]=\nsum_{k=0}^{\u-1}\Pr[\e[J]\cdot\pi_k=\e[\setofindexes]\mid K=k,Q_{\coalition}=(\oldPi\semicolon\vecQ[\j_1],\ldots,\vecQ[\j_{\t}])]\cdot\Pr[K=k\mid E_{\oldPi}]$.

Therefore, we have
\begin{align*}
	\Pr\bigl[I\mkern-0.75mu=\mkern-0.5mu\setofindexes\mkern-0.75mu\bigm\vert\mkern-0.75muE_{\oldPi}\mkern-0.25mu\bigr]\mkern-1mu
		&=\mkern-2mu\nsum[1.35]_{k=0}^{\u-1}\mkern-2mu\Pr[\e[J]\mkern-0.5mu\cdot\mkern-0.5mu\pi_k\mkern-0.5mu=\mkern-0.5mu\e[\setofindexes]\mkern-0.5mu\mid\mkern-0.5mu K\mkern-0.5mu=\mkern-0.5muk,E_{\oldPi}]\cdot\Pr[K=k\mid E_{\oldPi}]\\
		&=\mkern-2mu\nsum[1.35]_{k=0}^{\u-1}\mkern-2mu\Pr[\e[J]\mkern-0.5mu\cdot\mkern-0.5mu\pi_k\mkern-0.5mu=\mkern-0.5mu\e[\setofindexes]\mkern-0.5mu\mid\mkern-0.5mu K\mkern-0.5mu=\mkern-0.5muk,Q_{\coalition}=(\oldPi\semicolon\vecQ[\j_1],.,\vecQ[\j_{\t}])]\\
                &\cdot\Pr[K=k\mid E_{\oldPi}]\\
		&=\mkern-2mu\Pr\bigl[I=\setofindexes\bigm\vert Q_{\coalition}=(\oldPi\semicolon\vecQ[\j_1],\ldots,\vecQ[\j_{\t}])\bigr],
\end{align*}
Here, the second line utilizes the $\t$-privacy of secret shares, \smash{$(\x[1],\vecQ[1]),\ldots,(\x[\ell],\vecQ[\ell])$}, of the query vector received by $\ell$ servers, individually.
\end{proof}

\begin{corollary}\label{cor:uary}
The proposed private aggregate queries protocol is a $\t$-private IT-PIR as long as at least $\t+\u$ of $\ell$ servers respond correctly.
\end{corollary}

For $\u=1$, i.e., a single index of aggregate queries, no batching occurred, we have\emdash
\begin{corollary}\label{cor:uary}
The proposed private aggregate queries protocol is a $\t$-private IT-PIR as long as at least $\t+1$ of $\ell$ servers respond correctly.
\end{corollary}

\newpage
\section{Artifact Appendix}

In this appendix, we will provide a high-level guide for the usage of our codebase and for reproducing the results presented in our work. Our experiments were performed on a server with $256$GB of RAM and $5$ GPUs. Software installations required for using our codebase include PERCY++, Socket++, CUDA, NTL, and GMP. The artifact provided at \url{https://github.com/smhafiz/private_queries_it_pir/tree/v1.0.0} supports the research in the paper by facilitating practical, real-world implementations of our protocol and also provides the capability for benchmarking the protocol on different hardware configurations before actual implementation in live databases. \url{https://github.com/smhafiz/private_queries_it_pir/blob/v1.0.0/AE_Doc_Revised.pdf} provides a detailed instruction manual for using our code for benchmarking, validating our results, as well as implementation on other datasets. This instruction manual categorically links each of our major results, both benchmarking as well as case studies, to our artifact in the repository.

\subsection{Description \& Requirements}
\subsubsection{How to access}
The codebase for our artifact and the corresponding detailed instruction manual are available at the publicly available repository: \url{https://github.com/smhafiz/private_queries_it_pir/tree/v1.0.0}. Additionally, the AEC-approved artifact version is available at the following DOI link: \url{https://doi.org/10.5281/zenodo.10225325}.
\subsubsection{Hardware Requirements}
The experimentation server had $256$GiB RAM and $5$ GPU cores. The protocol can be implemented on servers with fewer hardware resources as well, with an expected impact on performance. However, it is recommended to use a server with a GPU.
\subsubsection{Software Requirements}
Experiments were performed on a Ubuntu 22.04 server with kernel-5.4.0-88-generic and GCC version 9.4.0. Additionally, the list of dependencies required for running the codes are:
\begin{itemize}
    \item GMP 6.3.0
    \item NTL 11.5.1
    \item Socket++
    \item Percy++ 1.0.0 
\end{itemize}
Further instructions for environment setup are available in our detailed instruction manual available at \url{https://github.com/smhafiz/private_queries_it_pir/blob/v1.0.0/AE_Doc_Revised.pdf}.

\subsubsection{Benchmark}
While the raw datasets used for our experiments are not part of the artifact, the indexes of aggregate queries for replicating our case study experiments are part of the artifact and stored in compressed column storage (CCS) format as outlined in our detailed instruction manual.

\subsection{Artifact Installation \& Configuration}

In order to configure the artifact, all the files and folders from the repository \url{https://github.com/smhafiz/private_queries_it_pir/tree/v1.0.0} need to be downloaded into the machine being used for running the experiments. Then, instructions enlisted at \url{https://github.com/smhafiz/private_queries_it_pir/blob/v1.0.0/AE_Doc_Revised.pdf} will have to be followed to complete the configuration of the environment. These instructions outline the installation of the dependencies and the corresponding configuration steps.

\subsection{Experiment Workflow}
The results section can be divided into two major sections: benchmarking and case studies. The setup and configuration mentioned in the previous section must precede all steps of either experimental workflow. The experimental workflow for the benchmarking results is straightforward, as we have standalone scripts for each benchmarking experiment for replicating them. For the case studies, the experimental flow is outlined in detail in the instruction manual at \url{https://github.com/smhafiz/private_queries_it_pir/blob/v1.0.0/AE_Doc_Revised.pdf}. It also includes a comprehensive example to help replicate the results and use the artifact in general. The workflow follows the sequence of subsections in the Case Studies section in the instruction manual. These are: 
\begin{itemize}
    \item Real-world Dataset collection (GPU not required)
    \item Scanning database and generating the IAQ matrices (GPU not required)
    \item Batching the IAQ matrices (GPU not required)
    \item Client query vector and IAQ sparse matrix multiplication time (ideally GPU required) 
    \item Server response generation time (GPU not required)
\end{itemize}

\subsection{Major Claims}
The major claims regarding the artifact are as follows:

\begin{itemize}
    \item (Claim 1): Our benchmarking experiments demonstrate the scalability of our protocol, and our artifact will help demonstrate similar trends. Also, provided testing is performed on similar hardware, our artifact will reliably reproduce the results of our benchmarking experiments.
    \item (Claim 2): Our protocol can be implemented in several real-world databases, as demonstrated by our case studies, and typically outperforms the baseline Goldberg's IT-PIR protocol. The artifact will be able to reiterate this claim and help reproduce the results of our case studies. The artifact can also be used with minor modifications for implementation on other datasets.
\end{itemize}

\subsection{Evaluation}

All the detailed setup, configuration steps, operational steps, and experimental procedures for the evaluation of our protocol through our artifact are provided at \url{https://github.com/smhafiz/private_queries_it_pir/blob/v1.0.0/AE_Doc_Revised.pdf}. Our experiments can be broadly categorized into the benchmarking experiments and the case studies. The document first enlists the necessary steps for reproducing the benchmarking experiments, linking each major result in our benchmarking section to experimentation instructions. Results can be reproduced if the same hardware configuration as ours is used. The document also provides a detailed set of instructions on how to perform our case study experiments. Under the assumption that the same hardware specifications are available during the evaluation of the artifact as ours, the time required for generating the case study results should be as indicated in the results, and most of the benchmarking experiments should be trivial (with the exception of Figure 5, which can take up to a few hours). Error bars for the vast majority of our experiments were found to be trivial, and that should be the case for any evaluation performed using our artifact.

\subsubsection{Preparation}

The steps outlined for preparing and configuring the software environment, as mentioned earlier and in the instruction manual, are to be performed first before any experiments can be reproduced.

\subsubsection{Execution}

The benchmarking results are obtained by running the individual scripts for the corresponding experiments. In order to execute the case study scripts, the codes associated with each section of the experimental workflow outlined earlier must be executed in sequence as per the detailed instructions in the manual. Please refer to the example at the end of the instruction manual to get a sense of the exact commands that need to be executed. 

\subsubsection{Results}

The results for the benchmarking experiments are stored in corresponding output files and can be plotted to replicate the graphs. For the case studies, most results are printed out into the terminal. Details on how to observe the results are available in the instruction manual.

\subsection{Customization}

Our artifact allows several kinds of customization. For the benchmarking experiments, benchmarks can be obtained over different ranges of values for parameters and for different fixed values of parameters by making small tweaks inside the shell script. The variable values can be altered to desired values or ranges of values to obtain alternative benchmarks. The case study results are specific to the datasets we used, but in order to use the artifact for other datasets, the steps in the instruction manual can be used to prepare the indexes of aggregate queries for other datasets and queries. Once they are stored in compressed column storage (CCS) format, the succeeding steps in the instruction manual should allow a user to obtain the desired server response times.

\subsection{Notes}

This appendix only provides a high-level guideline for using our artifact and understanding the experimental flow. Our instruction manual available with the artifact at \url{https://github.com/smhafiz/private_queries_it_pir/blob/v1.0.0/AE_Doc_Revised.pdf} is a more comprehensive and thorough walkthrough of how to use the artifact. It contains details regarding each step and outlines the exact instructions that need to be executed. We strongly recommend reading the manual for a more in-depth understanding of how to use our artifact to perform our experiments as well as extend it for further experiments with new datasets.

\end{document}